\documentclass[10pt,times,letter]{article}

\usepackage{amsmath, amsfonts,amsthm,amssymb,bm}
\usepackage{dsfont}

\usepackage[english]{babel}

\usepackage[normalem]{ulem}

\usepackage{latexsym,amssymb,amsfonts,graphicx}
\usepackage{amsmath,dsfont}
\usepackage{verbatim}
\usepackage{mathrsfs}
\usepackage{bm}
\usepackage{color}
\usepackage{epsfig}

\usepackage{url}
\usepackage{bm}
\usepackage{natbib}

\usepackage[colorlinks,linkcolor=blue,citecolor=blue,anchorcolor=blue]{hyperref}

\newcommand{\Px}{ \mathbb{P} }

\newcommand{\N}{ \mathbb{N} }
\newcommand{\Ex}{ \mathbb{E} }

\def\esssup_#1{\underset{#1}{\mathrm{ess\,sup\, }}}
\def\essinf_#1{\underset{#1}{\mathrm{ess\,inf\, }}}
\def\argmax_#1{\underset{#1}{\mathrm{arg\,max\, }}}
\def\argmin_#1{\underset{#1}{\mathrm{arg\,min\, }}}

\newcommand{\Fx}{\mathbb{F} }

\newcommand{\F}{\mathcal{F}}

\newcommand{\R}{\mathds{R}}
\newcommand{\taue}{\tau_{\varepsilon}^t}

\newcommand{\AZ}{{\bf(A$_{Z}$)}}
\newcommand{\Af}{{\bf(A$_{f}$)}}

\newtheorem{theorem}{Theorem}[section]

\numberwithin{equation}{section}

\newtheorem{proposition}[theorem]{Proposition}
\newtheorem{remark}[theorem]{Remark}
\newtheorem{lemma}[theorem]{Lemma}
\newtheorem{corollary}[theorem]{Corollary}

\allowdisplaybreaks[4]

\definecolor{Red}{rgb}{1.00, 0.00, 0.00}

\definecolor{DRed}{rgb}{0.5, 0.00, 0.00}

\definecolor{Blue}{rgb}{0.00, 0.00, 1.00}

\definecolor{Green}{rgb}{0.0, 0.4, 0.0}

\title{Optimal Tracking Portfolio with A Ratcheting Capital Benchmark}
\author{Lijun Bo \thanks{Email: lijunbo@ustc.edu.cn, School of Mathematics and Statistics, Xidian University, Xi'an, 710071, China. 
}\and
Huafu Liao \thanks{Email: mathuaf@nus.edu.sg, Department of Mathematics, National University of Singapore, Singapore 119076, Singapore.} \and
Xiang Yu \thanks{E-mail: xiang.yu@polyu.edu.hk, Department of Applied Mathematics, The Hong Kong Polytechnic University, Hung Hom, Kowloon, Hong Kong.}}
\date{\vspace{-5ex}}

\begin{document}
\maketitle
\begin{abstract}
This paper studies the finite horizon portfolio management by optimally tracking a ratcheting capital benchmark process. It is assumed that the fund manager can dynamically inject capital into the portfolio account such that the total capital dominates a non-decreasing benchmark floor process at each intermediate time. The tracking problem is formulated to minimize the cost of accumulated capital injection. We first transform the original problem with floor constraints into an unconstrained control problem, however, under a running maximum cost. By identifying a controlled state process with reflection, the problem is further shown to be equivalent to an auxiliary problem, which leads to a nonlinear Hamilton-Jacobi-Bellman (HJB) equation with a Neumann boundary condition. By employing the dual transform, the probabilistic representation and some stochastic flow analysis, the existence of the unique classical solution to the HJB equation is established. The verification theorem is carefully proved, which gives the complete characterization of the feedback optimal portfolio. The application to market index tracking is also discussed when the index process is modeled by a geometric Brownian motion.

\end{abstract}

\section{Introduction}

Portfolio allocation with benchmark performance has been an active research topic in quantitative finance, see some recent studies in \cite{Browne00}, \cite{Gaivoronski05}, \cite{YaoZZ06}, \cite{Strub18} and many others. The target benchmark is usually a prescribed capital process or a specific portfolio in the financial market, and the goal is to choose the portfolio in some risky assets to track the return or the value of the benchmark process. In practice, both professional and individual investors may measure their porfolio performance using different benchmarks, such as S$\&$P500 index, Goldman Sachs commodity index, special liability, inflation and exchange rates. The existing research mainly focuses on mathematical problems that minimize the difference between the controlled portfolio and the benchmark, which are formulated as either a linear quadratic control problem using the mean-variance analysis or a utility maximization problem at the terminal time. The present paper aims to enrich the study of optimal tracking by proposing a different tracking procedure and analyzing the associated control problem. We are particularly interested in the fund management when the fund manager can dynamically inject capital to keep the total fund capital above a specific non-decreasing benchmark process at each intermediate time. The control problem involves the regular portfolio control and the singular capital injection control together with American type floor constraints. The optimality is attained when the cost of the accumulated capital injection is minimized.

On the other hand, another well known optimal tracking problem in the literature is the monotone follower problem, see for instance \cite{KaraShreve84} and \cite{BE2008}, in which one needs to choose a monotone process as a singular control to closely track a given diffusion process such as a Brownian motion with drift. This paper investigates the opposite direction as we look for a regular control such that the controlled diffusion process can closely follow a given monotone process. Our mathematical problem is also motivated by some stochastic control problems with minimum guaranteed floor constraints, which are conventionally defined as utility maximization problems such that the controlled wealth processes dominate an exogenous target process at the terminal time or at each intermediate time. See some related studies among \cite{Karoui05}, \cite{KM06}, \cite{BMM10}, \cite{Giacinto11}, \cite{Sekine12}, \cite{Giacinto14} and \cite{CYZ20}, in which European type or American type floor constraints have been examined in various market models. In the aforementioned research, some typical techniques to handle the floor constraints are to introduce the option based portfolio or the insured portfolio allocation such that the floor constraints can be guaranteed. We instead reformulate the optimal tracking problem with dynamic floor constraints to a constraint-free stochastic control problem under a running maximum cost criterion, see Lemma \ref{first-transf}.

Stochastic control with a running maximum cost or a running maximum process is itself an interesting topic and attracted a lot of attention in the past decades, see \cite{BaIshii89}, \cite{Barr93}, \cite{BDR1994}, \cite{BPZ15} and \cite{Kroner18}, in which the viscosity solution approach plays the key role.
In contrast, we can take advantage of the specific payoff function and state processes to conclude the existence of a classical solution to the HJB equation. See also some recent control problems on optimal consumption in \cite{GHR20} and \cite{DLPY20}, in which the utility function depends on the running maximum of the control and the value function can be obtained explicitly. However, as opposed to \cite{GHR20} and \cite{DLPY20}, the running maximum of a controlled diffusion process appears in our finite horizon control problem and complicates the analysis. We choose to work with an auxiliary state process with reflection similar to \cite{Weerasinghe16} to reformulate the control problem again, which corresponds to a nonlinear HJB equation with a Neumann boundary condition. By using the heuristic dual transform, we can apply the probabilistic representation and some stochastic flow arguments to first establish the existence and uniqueness of the classical solution to the dual PDE. In fact, our primal value function is not strictly concave and the inverse transform needs to be carefully carried out in a restricted domain. To this end, we derive an explicit threshold for the initial wealth, beyond which the ratcheting benchmark process is dynamically superhedgeable by the portfolio process and no capital needs to be injected to catch up with the benchmark. This threshold facilitates the full characterization of the primal value function on the whole domain using some delicate continuity and convergence analysis based on the probabilistic representation results. The piecewise feedback optimal portfolio across different regions can also be derived and rigorously verified.

Although the primary focus of the present paper is to track a non-decreasing benchmark process, our approach and theoretical results can be applied in some market index tracking problems when the index process is not monotone. In Section \ref{sec:exam}, we present some examples when the index process follows a geometric Brownian motion, and show that the market index tracking problem can be transformed into an equivalent problem with a non-decreasing benchmark. In the model with infinite time horizon, the value function and the optimal portfolio can be obtained fully explicitly.

The rest of the paper is organized as follows. Section \ref{sec:model} introduces the model of the optimal tracking problem with a non-decreasing benchmark. To handle dynamic floor constraints, we consider an equivalent formulation with a running maximum cost. This problem is further transformed in Section \ref{sec:PDE} to an auxiliary one that leads to a nonlinear HJB equation with a Neumann boundary condition. Section \ref{sec:dual} studies the existence of a classical solution to the HJB equation using the dual transform and some probabilistic representation techniques. The feedback optimal portfolio and the proof of the verification theorem are given in Section \ref{sec:verify}. Section \ref{sec:exam} presents some illustrative examples when the index process follows a geometric Brownian motion. Section \ref{app:proof1} collects the proofs of some main results in previous sections. At last, conclusions are given in Section \ref{sec:conc}.

\section{Market Model and Problem Formulation}\label{sec:model}
Let $(\Omega, \mathbb{F}, \mathbb{P})$ be a filtered probability space, in which $\mathbb{F}=(\mathcal{F}_t)_{t\in [0,T]}$ satisfies the usual conditions. The process $(W^1,\ldots,W^d)$ is a $d$-dimensional Brownian motion adapted to $\mathbb{F}$. Let $T\in\R_+:=(0,\infty)$ be the finite terminal horizon. The financial market consists of $d$ risky assets and the price processes are described by, for $t\in[0,T]$,
\begin{align}
\frac{dS_t^i}{S_t^i}= \mu_idt+\sum_{j=1}^d\sigma_{ij}dW_t^{j},\ \ i=1,\ldots,d,\label{stockSDE}
\end{align}
with constant drift $\mu_i\in\R$ and constant volatility $\sigma_{ij}\in\R$ for $i,j=1,...,d$. It is assumed that the interest rate $r=0$ that amounts to the change of num\'{e}raire. From this point onwards, all processes including the wealth process and the benchmark process are defined after the change of num\'{e}raire.

For $t\in[0,T]$, let us denote $\theta_t^i$ the amount of wealth (as an $\Fx$-adapted process) that the fund manager allocates in asset $S^i=(S_t^i)_{t\in[0,T]}$ at time $t$. The self-financing wealth process under the control $\theta=(\theta_t^1,\ldots,\theta_t^d)_{t\in[0,T]}^{\top}$ is given by
\begin{align}\label{eq:wealth2}
V^{\theta}_t &=\textrm{v}+\int_0^t\theta_s^{\top}\mu ds+\int_0^t\theta_s^{\top}\sigma dW_s,\ \ \ t\in[0,T],
\end{align}
with the initial wealth $V_0^{\theta}=\textrm{v}\geq 0$, the return vector $\mu=(\mu_1,\ldots,\mu_d)^{\top}$ and the volatility matrix $\sigma=(\sigma_{ij})_{d\times d}$ that is assumed to be invertible (its inverse is denoted by $\sigma^{-1}$).

We consider the portfolio allocation by a fund manager that is to optimally track an exogenous non-decreasing capital benchmark process $A=(A_t)_{t\in[0,T]}$ taking the absolutely continuous form that
\begin{align}\label{nond-A}
A_t:=a+\int_0^t f(s,Z_s)ds,\quad t\in[0,T].
\end{align}
Here, $a\geq 0$ stands for the initial benchmark that the fund manager needs to track at time $t=0$. The function $f(\cdot,\cdot)$, representing the benchmark growth rate, is required to satisfy the condition:\\
\ \\
$(\mathbf{A}_f)$: the function $f:[0,T]\times\R\to \R_+$ is continuous and for $t\in[0,T]$, $f(t,\cdot)\in C^2(\R)$ with bounded first and second order derivatives.\\
\ \\
The stochastic factor process $Z=(Z_t)_{t\in[0,T]}$ in \eqref{nond-A} satisfies the SDE:
\begin{align}\label{factor-Z}
dZ_t =\mu_Z(Z_t) dt + \sigma_Z(Z_t)dW^{\gamma}_t,\ \ t\in[0,T],
\end{align}
with the initial value $Z_0=z\in\R$ and $W^{\gamma}=(W_t^{\gamma})_{t\in[0,T]}$ is a linear combination of the $d$-dimensional Brownian motion $(W^1,\ldots,W^d)$ with weights $\gamma=(\gamma_1,\ldots,\gamma_d)^{\top}\in[-1,1]^d$, which itself is a Brownian motion. We impose the following conditions on coefficients $\mu_Z(\cdot)$ and $\sigma_Z(\cdot)$ that:\\
\ \\
$(\mathbf{A}_Z)$: the coefficients $\mu_Z:\R\to\R$ and $\sigma_Z:\R\to\R$ belong to $C^2(\R)$ with bounded first and second order derivatives.

\begin{remark}
If $Z$ is an OU process or a geometric Brownian motion, the assumption {\AZ} clearly holds. The reasons for us to consider the non-decreasing benchmark process $A=(A_t)_{t\in[0,T]}$ in \eqref{nond-A} are twofold. Firstly, the process $A=(A_t)_{t\in[0,T]}$ may refer to some non-decreasing growth rate process after the change of num\'{e}raire, which generalizes the deterministic growth rate benchmark considered in \cite{YaoZZ06} for their optimal tracking problem. This non-decreasing benchmark process can also describe some consumer price index or higher eduction price index, which is observed to be non-decreasing over the past decades. That is, our portfolio management problem might be suitable to model some pension fund management or education savings fund management when the aim is to track some non-decreasing price index affected by the stochastic factor $Z=(Z_t)_{t\in[0,T]}$.  Secondly, one key step in our approach is the dual transform, see Section \ref{sec:dual}, which relies crucially on the concavity of the primal value function and the convexity of the solution to the dual PDE problem. The assumption $f>0$ facilitates the explicit characterization of the positive threshold defined in \eqref{eq:xi-bound} such that the primal value function is strictly concave when the wealth level is below this threshold. In addition, $f>0$ also guarantees that the solution to the dual PDE in Corollary \ref{coro:wellposehatv} is strictly convex in the interior domain such that the inverse transform is well defined. Despite that our main results hinge on the non-decreasing feature of the benchmark process, as shown in Section \ref{sec:exam}, the theoretical results are also applicable to some market index tracking problems when the index process follows a geometric Brownian motion. 
\end{remark}

Given the non-decreasing benchmark process $A$, an optimal tracking problem is considered that combines the portfolio control with another capital injection singular control together with dynamic floor constraints. To be precise, we assume that the fund manager can inject capital carefully to the fund account from time to time whenever it is necessary such that the total capital dynamically dominates the benchmark floor process $A$. That is, the fund manager optimally tracks the process $A$ by choosing the regular control $\theta$ as the dynamic portfolio in risky assets and the singular control $C=(C_t)_{t\in[0,T]}$ as the cumulative capital injection such that $C_t+V_t^{\theta}\geq A_t$ at each intermediate time $t\in[0,T]$. The goal of the optimal tracking problem is to minimize the expected cost of the discounted cumulative capital injection under American-type floor constraints that
\begin{align}\label{eq_prob_IBP}
u(a, \mathrm{v}, z):=\text{$\inf_{C,\theta} \Ex\left[C_0+\int_0^T e^{-\rho t}dC_t \right]$\ \ \ subject to\ \  $A_t \le C_t + V^{\theta}_t$ at each $t\in[0,T]$,}
\end{align}
where the constant $\rho\geq 0$ is the discount rate and $C_0=(a-\textrm{v})^+$ is the initial injected capital to match with the initial benchmark. This new tracking formulation \eqref{eq_prob_IBP} was initially developed with Martin Larsson and Johannes Ruf in a working paper for some more general index tracking problems.

\begin{remark}
Recall that we consider the model after the change of num\'{e}raire. The non-negative discount rate $\rho\geq 0$ is equivalent to the assumption that the discount rate in the original market dominates the interest rate before the change of num\'{e}raire. For the symmetric case, one can simply consider problem \eqref{eq_prob_IBP} with $\rho=0$.
\end{remark}

\begin{remark}
For a large initial wealth $\mathrm{v}\gg a$ and some special choices of $f(t,z)$ and $Z=(Z_t)_{t\in[0,T]}$, it is possible that the benchmark process $A=(A_t)_{t\in[0,T]}$ is dynamically superhedgeable by a portfolio in risky assets at each time $t\in[0,T]$. That is, there exists a portfolio $\theta^*$ such that $V_t^{\theta^*}\geq A_t$, for any $t\in[0,T]$. Then $C^*_t\equiv 0$ for any $t\in[0,T]$ is an admissible capital injection control and $(0, \theta^*)$ is an optimal control for the problem \eqref{eq_prob_IBP} and the value function $u(a,\mathrm{v},z)\equiv 0$. We will characterize the region for $\mathrm{v}$ explicitly in Remark \ref{Cequal0} such that there is no need to inject capital for the problem \eqref{eq_prob_IBP}.
\end{remark}
To tackle the problem \eqref{eq_prob_IBP} with dynamic floor constraints, our first step is to reformulate it based on the observation that, for a fixed control $\theta$, the optimal $C$ is always the smallest adapted right-continuous and non-decreasing process that dominates $A-V^{\theta}$. Let ${\cal U}$ be the set of regular $\Fx$-adapted control processes $\theta=(\theta_t)_{t\in[0,T]}$ such that \eqref{eq:wealth2} is well-defined. The following lemma gives an equivalent formulation of \eqref{eq_prob_IBP} and its proof is given in Section \ref{app:proof1}.
\begin{lemma}\label{first-transf}
For each fixed regular control $\theta$, the optimal singular control $C^*$ satisfies
\begin{align}\label{C-sing}
C^*_t =0\vee \sup_{s\in[0,t]}(A_s-V_s^{\theta}), \quad t\in [0,T].
\end{align}
The problem \eqref{eq_prob_IBP} with the American-type floor constraints $A_t\leq C_t+V_t^{\theta}$ for all $t\in[0,T]$, admits the equivalent formulation as a unconstrained control problem under a running maximum cost that
\begin{equation}\label{eq_orig_pb}
u(a,\mathrm{v}, z)=(a-\mathrm{v})^++\inf_{\theta\in\mathcal{U}}\ \Ex\left[ \int_0^T e^{-\rho t} d\left(0\vee \sup_{s\in[0,t]}(A_s-V_s^{\theta})\right)\right].
\end{equation}
\end{lemma}
\begin{remark}
To handle the running maximum term in the objective function, one can choose the monotone running maximum process as a controlled state process as in \cite{BDR1994}, \cite{Kroner18} to derive the HJB equation with a free boundary condition. One can also choose the distance between the underlying process and its running maximum as a reflected state process as in \cite{Weerasinghe16} and derive the HJB equation with a Neumann boundary condition. In the present paper, we follow the second method that allows us to prove the existence of a classical solution using the probabilistic representation results. Note that our problem mathematically differs from the one in \cite{Weerasinghe16} because we consider the control in both the drift and volatility of the state process together with a stochastic factor process affecting the benchmark capital.
\end{remark}

\section{Auxiliary Control Problem and HJB Equation}\label{sec:PDE}

In this section, we introduce a new controlled state process to replace the process $V^{\theta}=(V_t^{\theta})_{t\in[0,T]}$ given in \eqref{eq:wealth2} and formulate another auxiliary control problem. Let us first define the difference process $D_t:=A_t-V_t^{\theta}+\mathrm{v}-a$ with the initial value $D_0=0$. For any $x\geq 0$, we then consider its running maximum process $L=(L_t)_{t\in[0,T]}$ defined by
\begin{align}\label{eq:maxM}
L_t :=x\vee \sup_{s\in[0,t]}D_s\geq 0,\quad t\in [0,T],
\end{align}
with the initial value $L_0=x\geq 0$.

One can easily see that $(a-\mathrm{v})^+-u(a,\mathrm{v}, z)$ with $u(a,\mathrm{v},z)$ given in \eqref{eq_orig_pb} is equivalent to the auxiliary control problem
\begin{align}\label{eq:objectivefcn}
\sup_{\theta\in\mathcal{U}}\Ex\left[-\int_0^{T}e^{-\rho s}dL_s\right],
\end{align}
when we set the initial level $L_0=x=(\mathrm{v}-a)^+$. We can start to introduce our new controlled state process $X=(X_t)_{t\in[0,T]}$ for the problem \eqref{eq:objectivefcn}, which is defined as the reflected process $X_t:=L_t-D_t$ for $t\in [0,T]$ that satisfies the SDE, for $t\in[0,T]$,
\begin{align}\label{state-X}
X_t = -\int_0^tf(s,Z_s)ds +\int_0^t\theta_s^{\top}\mu ds+\int_0^t\theta_s^{\top}\sigma dW_s + L_t,
\end{align}
with the initial value $X_0=x\geq 0$. In particular, the running maximum process $L_t$ increases if and only if $X_t=0$, i.e., $L_t=D_t$. In view of ``the Skorokhod problem", it satisfies the representation that
\begin{align*}
L_t =x\vee \int_0^t {\bf1}_{\{X_s=0\}}dL_s,\quad t\in[0,T].
\end{align*}
We shall change the notation from $L_t$ to $L_t^X$ from this point onwards to emphasize its dependence on the new state process $X$ given in \eqref{state-X}. Moreover, the stochastic factor process $Z=(Z_t)_{t\in[0,T]}$ defined in \eqref{factor-Z}  is chosen as another state process.

For ease of presentation, we denote the domain $\mathcal{D}_T:=[0,T]\times\R\times[0,\infty)$. Let ${\cal U}_t$ be the set of admissible controls taking the feedback form as $\theta_s=\theta(s,Z_s,X_s)$ for $s\in[t,T]$, where $\theta:\mathcal{D}_T\to\R^n$ is a measurable function such that the following reflected SDE has a weak solution that
\begin{align}\label{eq:reflectedZA}
X_s = -\int_t^s f(r,Z_r)dr+ \int_t^s\theta(r,Z_r,X_r)^{\top}\mu dr+\int_t^s\theta(r,Z_r,X_r)^{\top}\sigma dW_r + L_s^{X},\quad s\in[t,T],
\end{align}
with $X_t=x\geq 0$. Here, $L_s^X=x\vee \int_t^s{\bf1}_{\{X_r=0\}}dL_r^X$ is a continuous, non-negative and non-decreasing process, which increases only when the state process $(X_s)_{s\in[t,T]}$ hits the level $0$. For $(t,z,x)\in\mathcal{D}_T$, the dynamic version of the auxiliary problem \eqref{eq:objectivefcn} is given by
\begin{align}\label{eq:valuefcnA}
{\rm w}(t,z,x):=\sup_{\theta\in{\cal U}_t}J(\theta;t,z,x):=\sup_{\theta\in{\cal U}_t}\Ex_{t,z,x}\left[-\int_t^Te^{-\rho s}dL_s^X\right],
\end{align}
where $\Ex_{t,z,x}[~\cdot~]:=\Ex[~\cdot~|Z_t=z,X_t=x]$ denotes the conditional expectation and the underlying state processes $(Z_s)_{s\in[t,T]}$ and $(X_s)_{s\in[t,T]}$ are given in \eqref{factor-Z} and \eqref{eq:reflectedZA} respectively.

It is important to note the equivalence that ${\rm w}(0,z,(\mathrm{v}-a)^+)=(a-\mathrm{v})^+-u(a,\mathrm{v},z)$, i.e., we have

\begin{equation*}
u(a,\mathrm{v},z)=\left\{
\begin{aligned}
&\displaystyle a-\mathrm{v}-{\rm w}(0,z,0), & & \mbox{if } a\geq \mathrm{v},\\
&\displaystyle-{\rm w}(0,z,\mathrm{v}-a),  & & \mbox{if } a<\mathrm{v},
\end{aligned}
\right.
\end{equation*}
where $u(a,v,z)$ is the value function of the original optimal tracking problem defined by \eqref{eq_prob_IBP}, and $a$ (resp. $\mathrm{v}$) represents the initial benchmark level (resp. the initial wealth). We now mainly focus on the auxiliary control problem \eqref{eq:valuefcnA} and seek to obtain its optimal portfolio in a feedback form.

To simplify the associated HJB equation, let us consider the function
\begin{align}\label{PDE-sol}
v(t,z,x):=e^{\rho t}{\rm w}(t,z,x).
\end{align}
The next result gives some preliminary properties of the value function ${\rm w}$ on ${\cal D}_T$ defined in \eqref{eq:valuefcnA}. The proof is standard by following the solution representation of ``the Skorokhod problem" and it is hence omitted.
\begin{lemma}\label{lem:Lipv}
For $(t,z,x)\in{\cal D}_T$, the value function $v(t,z,x)$ defined by \eqref{eq:valuefcnA} is non-decreasing in $x\geq0$. Moreover, for all $(t,z)\in[0,T]\times\R$, we have
\begin{align*}
\left|{\rm w}(t,z,x_1)-{\rm w}(t,z,x_2)\right|\leq e^{-\rho t}\left| x_1-x_2\right|,\ \ \text{for all}\ x_1,x_2\geq0.
\end{align*}
It follows that
\begin{align*}
\left|v(t,z,x_1)-v(t,z,x_2)\right|\leq \left| x_1-x_2\right|,\ \ \text{for all}\ x_1,x_2\geq0.
\end{align*}
\end{lemma}

\begin{remark}\label{rem:derivzbound}
For $(t,z)\in[0,T]\times\R$, if $x\to v(t,z,x)$ is $C^1([0,\infty))$, Lemma~\ref{lem:Lipv} implies that the important range that $0\leq v_x(t,z,x)\leq 1$ for all $(t,z,x)\in{\cal D}_T$. Hereafter, we use $v_x$, $v_t$, $v_{xx}$, $v_{zx}$ and $v_{zz}$ to denote the (first, second order or mixed) partial derivatives of the value function $v$ with respect to its arguments, if exist.

\end{remark}

By some heuristic arguments, we can show that $v$ defined in \eqref{PDE-sol} satisfies the HJB equation:
\begin{align}\label{eq:HJBA}
\begin{cases}
\displaystyle v_t+\sup_{\theta\in\R^n} \left[v_x\theta^{\top}\mu+\frac{v_{xx}}{2}\theta^{\top}\sigma\sigma^{\top}\theta+v_{xz}\sigma_Z(z)\theta^{\top}\sigma\gamma\right]\\
\displaystyle\qquad+v_z\mu_Z(z)+v_{zz}\frac{\sigma_Z^2(z)}{2}-f(t,z)v_x=\rho v,\quad (t,z,x)\in[0,T]\times\R\times\R_+; \\ \\
\displaystyle v(T,z,x)=0,\quad \forall~(z,x)\in\R\times[0,\infty);\\ \\
\displaystyle v_x(t,z,0)=1,\quad \forall~(t,z)\in[0,T]\times\R.
\end{cases}
\end{align}
Here, the Neumann boundary condition $v_x(t,z,0)=1$ stems from the martingale optimality condition because the process $L_s^X$ increases whenever the process $X_s$ visits the value $0$ for $s\in[t,T]$. Suppose $v_{xx}<0$ on $[0,T)\times\R\times\R_+$, the feedback optimal control determined by \eqref{eq:HJBA} is obtained by
\begin{align}\label{eq:thetastar}
\theta^*(t,z,x)=-(\sigma\sigma^{\top})^{-1}\frac{v_x(t,z,x)\mu+v_{xz}(t,z,x)\sigma_Z(z)\sigma\gamma}{v_{xx}(t,z,x)},\ \ \ (t,z,x)\in\mathcal{D}_T.
\end{align}
Plugging \eqref{eq:thetastar} into the HJB equation \eqref{eq:HJBA}, we have for $(t,z,x)\in[0,T)\times\R\times\R_+$ that

\begin{align}\label{eq:HJBA1}
\displaystyle v_t-\rho v-\alpha \frac{v_x^2}{v_{xx}}+\frac{\sigma_Z^2(z)}{2}\left(v_{zz}-\frac{v_{xz}^2}{v_{xx}}\right)-\phi(z)\frac{v_xv_{xz}}{v_{xx}}+
\mu_Z(z)v_z-f(t,z)v_x=0,
\end{align}
where the coefficients are given by
\begin{align}\label{eq:alpha}
\alpha:=\frac{1}{2}\mu^{\top}(\sigma\sigma^{\top})^{-1}\mu,\qquad 
\phi(z):=\sigma_Z(z)\mu^{\top}(\sigma\sigma^{\top})^{-1}\sigma\gamma,\quad z\in\R.
\end{align}

Note that the HJB equation \eqref{eq:HJBA} is fully nonlinear. To study the existence of a classical solution to \eqref{eq:HJBA}, we will first apply the heuristic dual transform to linearize the original HJB equation~\eqref{eq:HJBA} and {then} establish the existence and uniqueness of a classical solution to the dual PDE using the probabilistic representation and stochastic flow analysis in the next section.

\section{Dual Transform and Probabilistic Representation}\label{sec:dual}
To employ the dual transform, let us first assume that the value function $v$ satisfies $v\in C^{1,2,2}([0,T)\times\R\times[0,\infty))\cap C(\mathcal{D}_T)$ and $v_{xx}<0$ on $[0,T)\times\R\times\R_+$, which will be discussed and verified in detail later, see Section \ref{sec:verify}. For $(t,z,y)\in[0,T]\times\R\times\R_+$, the dual transform is only applied with respect to $x$ that
\begin{align}\label{eq:conjugate}
  \hat{v}(t,z,y):=\sup_{x>0}\{v(t,z,x)-xy\}\ \ \text{and}\ \ x^*(t,z,y):=v_x(t,z,\cdot)^{-1}(y),
\end{align}
where $y\mapsto v_x(t,z,\cdot)^{-1}(y)$ denotes the inverse function of $x\mapsto v_x(t,z,x)$, and $x^*=x^*(t,z,y)$ in \eqref{eq:conjugate} satisfies the equation:
\begin{align}\label{eq:conjugate1}
v_x(t,z,x^*)=y,\ \ (t,z)\in[0,T]\times\R.
\end{align}
On the other hand, in view of Lemma~\ref{lem:Lipv} and Remark~\ref{rem:derivzbound}, the variable $y$ in fact only takes values in the set $(0,1)$. It follows by \eqref{eq:conjugate} that, for all $(t,z,y)\in[0,T]\times\R\times(0,1)$,
\begin{align}\label{eq:conjugate2}
  \hat{v}(t,z,y)=v(t,z,x^*)-x^*y.
\end{align}
Taking the derivative with respect to $y$ on both sides of \eqref{eq:conjugate2} and \eqref{eq:conjugate1}, we deduce that
\begin{align}\label{eq:hatvx-zstar}
  \hat{v}_y(t,z,y)=v_x(t,z,x^*)x^*_y-x^*_yy-x^*=yx^*_y-x^*_yy-x^*=-x^*,
\end{align}
and also $v_{xx}(t,z,x^*)x^*_y=1$ that yields $x^*_y=\frac{1}{v_{xx}(t,z,x^*)}$. Because of \eqref{eq:hatvx-zstar}, we can obtain that
\begin{align}\label{eq:conjugate4}
  \hat{v}_{yy}(t,z,y)=-x^*_{y}=-\frac{1}{v_{xx}(t,z,x^*)},\qquad x_z^*=-\frac{v_{xz}(t,z,x^*)}{v_{xx}(t,z,x^*)}.
\end{align}
It follows by \eqref{eq:conjugate1} and \eqref{eq:conjugate2} that
\begin{align}\label{eq:conjugate3}
&\hat{v}_t(t,z,y)=v_t(t,z,x^*),\quad \hat{v}_z(t,z,y)=v_z(t,z,x^*),\\
&\hat{v}_{zz}(t,z,y)=v_{zz}(t,z,x^*)-\frac{v_{xz}(t,z,x^*)^2}{v_{xx}(t,z,x^*)}.
\end{align}
Moreover, by the second equality in \eqref{eq:conjugate4} and \eqref{eq:conjugate3}, we further have that
\begin{align}\label{eq:hatvyztyzstar}
\hat{v}_{yz}(t,z,y)=v_{xz}(t,z,x^*)x^*_y=\frac{v_{xz}(t,z,x^*)}{v_{xx}(t,z,x^*)}.
\end{align}
By virtue of \eqref{eq:HJBA1} and \eqref{eq:conjugate2}, it holds that
\begin{align}\label{eq:HJBA100}
&v_t(t,z,x^*)-\rho v(t,z,x^*)-\alpha \frac{v_x(t,z,x^*)^2}{v_{xx}(t,z,x^*)}-\frac{\sigma_Z^2(z)}{2}\frac{v_{xz}(t,z,x^*)^2}{v_{xx}(t,z,x^*)}-\phi(z)\frac{v_x(t,z,x^*)v_{xz}(t,z,x^*)}{v_{xx}(t,z,x^*)}\nonumber\\
&\qquad+
\mu_Z(z)v_z(t,z,x^*)+\frac{\sigma_Z^2(z)}{2}v_{zz}(t,z,x^*)-f(t,z)v_x(t,z,x^*)=0.
\end{align}
Plugging \eqref{eq:conjugate1}, \eqref{eq:conjugate4}, \eqref{eq:conjugate3} and \eqref{eq:hatvyztyzstar} into \eqref{eq:HJBA100}, we can derive that, for $(t,z,y)\in[0,T)\times\R\times(0,1)$,
\begin{align}\label{eq:frenchelHJB1}
&\hat{v}_t(t,z,y)-\rho\hat{v}(t,z,y)+\rho y\hat{v}_y(t,z,y)+\alpha y^2\hat{v}_{yy}(t,z,y)+\mu_Z(z)\hat{v}_z(t,z,y)+\frac{\sigma_Z^2(z)}{2}\hat{v}_{zz}(t,z,y)\nonumber\\
&\qquad-\phi(z)y\hat{v}_{yz}(t,z,y)-f(t,z)y=0.
\end{align}

In addition, the terminal condition ${v}(T,z,x)=0$ of the HJB equation \eqref{eq:HJBA} then implies that
\begin{align}\label{eq:terminalhatv}
  \hat{v}(T,z,y)=\sup_{x>0}\{v(T,z,x)-xy\}=\sup_{x>0}\{-xy\}=0,\ \ \ (z,y)\in\R\times(0,1).
 \end{align}
Note that $x^*_{y}=\frac1{v_{xx}(t,z,x^*)}<0$, and for each $(t,z)\in[0,T]\times\R$, the map $y\mapsto x^*(t,z,y):=v_x(t,z,\cdot)^{-1}(y)$ is one to one. By the Neumann boundary condition of the HJB equation \eqref{eq:HJBA}, we deduce from \eqref{eq:conjugate1} that $v_x(t,z,0)=1$ and $x^*(t,z,1)=0$. Therefore, in view of \eqref{eq:hatvx-zstar}, for all $(t,z)\in[0,T]\times\R$, we have
\begin{align}\label{eq:Neumann-boundary-hatv}
  \hat{v}_y(t,z,1)=-x^*(t,z,1)=0.
\end{align}
In summary, the HJB equation~\eqref{eq:HJBA} can be transformed into the linear dual PDE of $\hat{v}$ that
\begin{align}\label{eq:hatvHJB1}
\begin{cases}
\displaystyle \hat{v}_t+\alpha y^2\hat{v}_{yy}+\rho y\hat{v}_y-\rho\hat{v}-\phi(z)y\hat{v}_{yz}+\mu_Z(z)\hat{v}_z+\frac{\sigma_Z^2(z)}{2}\hat{v}_{zz}\\
\qquad\quad-f(t,z)y=0,\quad (t,z,y)\in[0,T)\times\R\times(0,1);\\ \\
\displaystyle \hat{v}(T,z,y)=0,\quad \forall~(z,y)\in\R\times[0,1];\\ \\
\displaystyle \hat{v}_y(t,z,1)=0,\quad \forall~(t,z)\in[0,T]\times\R.
\end{cases}
\end{align}

We next study the existence and uniqueness of the classical solution to the problem~\eqref{eq:hatvHJB1} with the extra condition that $\hat{v}_{yy}\geq0$ on $[0,T)\times\R\times(0,1)$ using the probabilistic representation approach. To this end, for $(t,z,u)\in\mathcal{D}_T$, let us define the function
\begin{align}\label{eq:dualtransfromh}
h(t,z,u) := -\Ex\left[\int_t^T e^{-\rho s}f(s,M_s^{t,z})e^{-R_s^{t,u}}ds\right],
\end{align}
where the process $(M_s^{t,z})_{s\in[t,T]}$ with $(t,z)\in[0,T]\times\R$ satisfies the SDE, for $s\in[t,T]$,
\begin{align}\label{eq:Feynman-Kac-Y}
M_s^{t,z}=z + \int_t^s \mu_Z(M_r^{t,z})dr+\varrho\int_t^s\sigma_Z(M_r^{t,z})dB^1_r+\sqrt{1-\varrho^2}\int_t^s\sigma_Z(M_r^{t,z})dB^2_r.
\end{align}
The processes $B^1=(B^1_t)_{t\in[0,T]}$ and $B^2=(B^2_t)_{t\in[0,T]}$ are two standard Brownian motions with a specific correlation coefficient
\begin{align}\label{eq:correlationcoeffrho}
\varrho := \frac{(\sigma^{-1}\mu)^\top}{\left|\sigma^{-1}\mu\right|}\gamma.
\end{align}
Moreover, the process $(R_s^{t,u})_{s\in[t,T]}$ with $(t,u)\in[0,T]\times[0,\infty)$ is a reflected Brownian motion with drift defined by
\begin{align}\label{eq:driftBM}
R_s^{t,u} := u+ \sqrt{2\alpha}\int_t^sdB^1_r +\int_t^s(\alpha-\rho)dr + \int_t^s dL_r^{t,R}\geq0,\quad s\in[t,T],
\end{align}
where $[t,T]\ni s\mapsto L_s^{t,R}$ is a continuous and non-decreasing process that increases only on $\{s\in[t,T];~R_s^{t,u}=0\}$ with $L_t^{t,R}=0$. By the solution representation of ``the Skorokhod problem", we obtain that, for $(s,u)\in[t,T]\times[0,\infty)$,
\begin{align}\label{eq:Lskorokhodrep}
L_s^{t,R}
&=0\vee\left\{-u+\max_{r\in[t,s]}\left[-\sqrt{2\alpha}(B^1_r-B^1_t)-(\alpha-\rho)(r-t)\right]\right\}.
\end{align}

It follows from assumptions {\Af} and {\AZ} that, for all $(t,z,u)\in\mathcal{D}_T$, 
\begin{align}\label{eq:estidualtransfromh}
\left|h(t,z,u)\right| &= \Ex\left[\int_t^T e^{-\rho s}f(s,M_s^{t,z})e^{-R_s^{t,u}}ds\right]\leq C\Ex\left[\int_t^T e^{-\rho s}(1+|M_s^{t,z}|)ds\right]\nonumber\\
&\leq C(T-t)+C(T-t)\Ex\left[\sup_{s\in[t,T]}|M_s^{t,z}|\right]\leq C(T-t)(1+|z|),
\end{align}
for some constant $C=C_f>0$. Hence, the function $h$ given in \eqref{eq:dualtransfromh} is well-defined. We next study the regularity of the function $h$ defined in \eqref{eq:dualtransfromh} in the next result, and its proof is reported in Section~\ref{app:proof1}.
\begin{proposition}\label{prop:smoothh}
Let assumptions {\Af} and {\AZ} hold. We have that $h\in C^{1,2,2}(\mathcal{D}_T)$. Moreover, for $(t,z,u)\in\mathcal{D}_T$, we get
\begin{align}\label{eq:derihu}
h_u(t,z,u)&=\Ex\left[\int_t^Te^{-\rho s}f(s,M_s^{t,z})e^{-R_s^{t,u}}\mathbf{1}_{\left\{\max_{r\in[t,s]}[-\sqrt{2\alpha}(B^1_r-B^1_t)-(\alpha-\rho)(r-t)]\leq u\right\}}ds\right]\notag\\
  &=\Ex\left[\int_t^{\tau_{u}^t\wedge T} e^{-\rho s}f(s,M_s^{t,z})e^{-R_s^{t,u}}ds\right],
\end{align}
where $\tau_{u}^t:=\inf\{s\geq t;~-\sqrt{2\alpha}(B^1_s-B^1_t)-(\alpha-\rho)(s-t)=u\}$ (we assume $\inf\emptyset=+\infty$ by convention).
\end{proposition}
Building upon Proposition~\ref{prop:smoothh}, we have the following important result and its proof is given in Section \ref{app:proof1}.
\begin{theorem}\label{thm:hsolvePDE}
Suppose that {\Af} and {\AZ} hold. The function $h$ defined in \eqref{eq:dualtransfromh} solves the Neumann boundary problem:
\begin{align}\label{eq:hHJB2}
\begin{cases}
\displaystyle h_t+\alpha h_{uu}+(\alpha-\rho)h_u+\phi(z)h_{uz}+\mu_Z(z)h_{z}+\frac{\sigma_Z^2(z)}{2}h_{zz}\\
\qquad\qquad=f(t,z)e^{-u-\rho t},\quad (t,z,u)\in[0,T)\times\R\times\R_+;\\ \\
\displaystyle h(T,z,u)=0,\quad \forall~(z,u)\in\R\times[0,\infty);\\ \\
\displaystyle h_u(t,z,0)=0,\quad \forall~(t,z)\in[0,T]\times\R.
\end{cases}
\end{align}
On the other hand, if a function $h$ defined on $\mathcal{D}_T$ with a polynomial growth is a classical solution of the Neumann boundary problem \eqref{eq:hHJB2}, then $h$ has the representation~\eqref{eq:dualtransfromh}.
\end{theorem}

\begin{remark}\label{rem:hu+hpositive}
Under assumptions {\Af} and {\AZ}, it follows by \eqref{eq:derihu} and \eqref{eq:Rephuu} in Section~\ref{app:proof1} that, for all $(t,z,u)\in\mathcal{D}_T$,
\begin{align}\label{eq:h+hugeq0}
h_u(t,z,u)+h_{uu}(t,z,u)=\Ex\left[e^{-\rho \tau^t_u}f(\tau^t_u,M_{\tau^t_u}^{t,z})\Gamma(\tau^t_u)+2\int_t^{\tau^t_u}e^{-\rho s}f(s,M_s^{t,z})e^{-R_s^{t,u}}ds\right],
\end{align}
where the stopping time $\tau_u^t$ is given in Proposition~\ref{prop:smoothh} and the function $\Gamma(t)$ for $t\in[0,T]$ is given by \eqref{eq:Gammat} in Section~\ref{app:proof1}. Note that $f>0$ in the assumption~{\Af} guarantees that $h_{uu}+h_u>0$ for $(t,z,u)\in[0,T)\times\R\times[0,\infty)$, which implies that $\hat{v}(t,z,y)$ in \eqref{hat-v} is strictly convex in $y\in(0,1]$.
\end{remark}

The well-posedness of the problem \eqref{eq:hatvHJB1} is now given in the next result.
\begin{corollary}\label{coro:wellposehatv}
Let assumptions of Theorem \ref{thm:hsolvePDE} hold. The problem \eqref{eq:hatvHJB1} admits a unique classical solution $\hat{v}$ such that for $(t,z,y)\in[0,T]\times \R\times (0,1]$,
\begin{align}
  \big|\hat{v}(t,z,y)\big|\leq C(1+|z|^p+|\ln y|^p),\quad\text{for some }p>1,
\end{align}
and the function
\begin{align}\label{hat-v}
  h(t,z,u):=e^{-\rho t}\hat{v}(t,z,e^{-u}),\quad(t,z,u)\in\mathcal{D}_T,
\end{align}
has the probabilistic representation~\eqref{eq:dualtransfromh}. Moreover, for each $(t,z)\in[0,T)\times\R$, the solution $(0,1]\ni y\mapsto\hat{v}(t,z,y)$ is strictly convex.
\end{corollary}

\section{Optimal Portfolio and Verification Theorem}\label{sec:verify}
Corollary~\ref{coro:wellposehatv} gives the existence and uniqueness of the classical solution $\hat{v}(t,z,y)$ for $(t,z,y)\in[0,T]\times\R\times(0,1]$ to the dual PDE \eqref{eq:hatvHJB1} that is strictly convex in $y\in (0,1]$. We next recover the classical solution $v(t,z,x)$ of the primal HJB equation~\eqref{eq:HJBA1} via $\hat{v}(t,z,y)$ using the inverse transform and prove the verification theorem of the primal stochastic control problem \eqref{eq:valuefcnA}.
\begin{theorem}[Verification theorem] \label{thm:verificationthem} Let assumptions {\Af} and {\AZ} hold. We have that:
\begin{itemize}
\item[{\bf(i)}] The primal HJB equation~\eqref{eq:HJBA1} admits a solution $v\in C^{1,2,2}([0,T)\times\R\times[0,\infty))\cap C(\mathcal{D}_T)$. Moreover, for $(t,z,x)\in\mathcal{D}_T$, the solution $v$ of HJB~equation~\eqref{eq:HJBA1} can be written by
\begin{align}\label{eq:solution-prime-hjb}
v(t,z,x)=\begin{cases}
\displaystyle \inf_{y\in(0,1]}\{\hat{v}(t,z,y)+xy\},\quad\text{if}~(t,z,x)\in{\cal O}_T\text{ or }x=0,\\ \\
\displaystyle  \qquad\quad 0,\quad\text{if}~(t,z,x)\in{\cal O}^c_T\cap\mathcal{D}_T.
\end{cases}
\end{align}
Here, the region ${\cal O}_T$ in \eqref{eq:solution-prime-hjb} is given by
\begin{align}\label{eq:OT}
  {\cal O}_T:=\left\{(t,z,x)\in[0,T)\times\R\times\R_+;~x\in(0,\xi(t,z))\right\},
\end{align}
where the function $\xi(t,z)$ with $(t,z)\in[0,T)\times\R$ is defined by
\begin{align}\label{eq:xi-bound}
  \xi(t,z):=\Ex\left[\int_t^Te^{-\rho(s-t)}f(s,M_s^{t,z})e^{\sqrt{2\alpha}(B^1_s-B^1_t)+(\alpha-\rho)(s-t)}ds\right],
\end{align}
and the process $(M_s^{t,z})_{s\in[t,T]}$ with $(t,z)\in[0,T]\times\R$ is the strong solution of SDE~\eqref{eq:Feynman-Kac-Y}. As $f(\cdot, \cdot)>0$ in \eqref{nond-A}, we get $\xi(\cdot, \cdot)>0$ so that ${\cal O}_T\neq \emptyset$. Here, for $(t,z,y)\in[0,T]\times\R\times(0,1]$, the function $\hat{v}(t,z,y)=e^{\rho t}h(t,z,-\ln y)$ solves the dual PDE \eqref{eq:hatvHJB1} with Neumann boundary condition and $\hat{v}(t,z,y)$ is strictly convex in $y\in(0,1]$.

\item[{\bf(ii)}]  For $(t,z,x)\in\mathcal{D}_T$, let us define the feedback control function by
\begin{align}\label{eq:verifioptimalstrategy}
\theta^*(t,z,x):=\begin{cases}
  \displaystyle -(\sigma\sigma^{\top})^{-1}\frac{v_x(t,z,x)\mu+v_{xz}(t,z,x)\sigma_Z(z)\sigma\gamma}{v_{xx}(t,z,x)},\quad\text{if}~(t,z,x)\in{\cal O}_T\text{ or }x=0, \\ \\
 \displaystyle -\mu(\sigma\sigma^\top)^{-1}\lim_{y\downarrow0}y\hat{v}_{yy}(t,z,y)\qquad\qquad\text{if}~(t,z,x)\in{\cal O}^c_T\cap\mathcal{D}_T.\\
 \displaystyle~+(\sigma\sigma^\top)^{-1}\sigma_Z(z)\sigma\gamma\lim_{y\downarrow0}\hat{v}_{yz}(t,z,y),
\end{cases}
\end{align}
Given the processes $(Z,X)=(Z_t,X_t)_{t\in[0,T]}$ in \eqref{eq:reflectedZA}, we define $\theta_t^*:=\theta^*(t,Z_t,X_t)$ for $t\in[0,T]$. Then $\theta^*=(\theta_t^*)_{t\in[0,T]}\in{\cal U}_t$ is an optimal strategy. Moreover, for all $\theta\in{\cal U}_t$, it holds that $J(\theta;t,z,x)\leq e^{-\rho t}v(t,z,x)={\rm w}(t,z,x)$, where $(t,z,x)\in[0,T)\times\R\times[0,\infty)$.
\end{itemize}
\end{theorem}

\begin{remark}\label{rem:Toptimalstrategy}

We explain here the role of the function $\xi(t,z)$ defined by \eqref{eq:xi-bound} in Theorem~\ref{thm:verificationthem}. In fact, for $(t,z)\in[0,T)\times\R$ and $x\geq\xi(t,z)$, it follows from Theorem~\ref{thm:verificationthem}-{\bf(i)} that the value function ${\rm w}(t,z,x)=0$. Then, by Theorem~\ref{thm:verificationthem}-{\bf(ii)}, we have that, for the strategy $\theta^*\in{\cal U}_t$ given by \eqref{eq:verifioptimalstrategy},
\begin{align*}
  \Ex_{t,z,x}\left[-\int_t^Te^{-\rho s}dL_s^{X^*}\right]=0,
\end{align*}
where the process $(L_s^{X^*})_{s\in[t,T]}$ is the reflected term of the process $(X^*_s)_{s\in[t,T]}$ in \eqref{eq:reflectedZA} with $\theta$ replaced by $\theta^*$. It follows from  integration by parts that $e^{-\rho T}L_T^{X^*}+\rho\int_t^Te^{-\rho s}L_s^{X^*}ds=x$, $\Px$-a.s.
and hence $L^{X^*}_T=L_t^{X^*}=x$, $\Px$-a.s. because $\xi(t,z)>0$ for $(t,z)\in[0,T)\times\R$. 

Therefore, with the strategy $\theta^*\in{\cal U}_t$, the non-negative process $X^*$ is given by
\begin{align*}
X_s^* = x- \int_t^s f(r,Z_r)dr+ \int_t^s(\theta_r^*)^{\top}\mu dr+\int_t^s(\theta_r^*)^{\top}\sigma dW_r.
\end{align*}
On the other hand, for $0\leq x<\xi(t,z)$, we have that $v_x(t,z,x)>0$ and hence ${\rm w}(t,z,x)<0$. This implies that, for this initial value $x$ at time $t$, the reflected term $L_s^{X^*}$ is strictly increasing in $s\in[t,T)$ with a positive probability.
\end{remark}

\begin{remark}\label{Cequal0}
Recall the equivalence that $u(a,\mathrm{v},z)=-{\rm w}(0,z,\mathrm{v}-a)$, $\mathrm{v}>a$, where $u(a,\mathrm{v},z)$ is the value function of the original optimal tracking problem \eqref{eq_prob_IBP}. By Remark \ref{rem:Toptimalstrategy}, if the initial wealth $\mathrm{v}$ is sufficiently large such that $\mathrm{v}-a>\xi{(0,z)}$ for the given $f(\cdot,\cdot)$ and $\mu_Z(\cdot)$, $\sigma_Z(\cdot)$, we can conclude that $u(a,\mathrm{v},z)=0$ and the optimal singular control $C_t^*\equiv 0$ for $t\in[0,T]$. Therefore, $A_t$ is dynamically superhedgeable that $A_t\leq V_t^{\theta^*}$ for $t\in[0,T]$.
\end{remark}

\begin{proof}[Proof of Theorem~\ref{thm:verificationthem}] We first show~{\bf(i)}. By the assumption {\Af}, we have $\xi(t,z)>0$, $(t,z)\in[0,T)\times\R$, according to its definition. Moreover, thanks to the probabilistic representation of derivatives of $h$ in the proof of Proposition~\ref{prop:smoothh}, we have that
\begin{align}\label{eq:hatderivatives}
\begin{cases}
\displaystyle \hat{v}_y(t,z,1)  =-e^{\rho t}h_u(t,z,0)=0,\quad \hat{v}_y(T,z,y)=y^{-1}e^{\rho t}h_u(T,z,-\ln y)=0,~~z\in(0,\infty),\\ \\
\displaystyle \hat{v}_{yy}(t,z,e^{-u})=e^{\rho t+2u}(h_u(t,z,u)+h_{uu}(t,z,u))\geq0,~u\in[0,\infty),~(``>" \text{holds~for}~t\in[0,T)),\\ \\
\displaystyle  \lim_{y\to0}\hat{v}_y(t,z,y)=\lim_{u\to+\infty}\hat{v}_y(t,z,e^{-u})=-\lim_{u\to+\infty}e^{\rho t+u}h_u(t,z,u)=-\xi(t,z),\\ \\
\displaystyle \lim_{y\to0}\hat{v}_{yy}(t,z,y)=\lim_{u\to+\infty}\hat{v}_{yy}(t,z,e^{-u})=\lim_{u\to+\infty}e^{\rho t+2u}(h_u(t,z,u)+h_{uu}(t,z,u))=+\infty,\\ \\
\displaystyle \lim_{y\to0}\hat{v}_{zz}(t,z,y)=\lim_{u\to+\infty}e^{\rho t}h_{zz}(t,z,u)=0,\quad  \lim_{y\to0}|\hat{v}_{yz}(t,z,y)|= \lim_{u\to+\infty}e^{\rho t+u}|h_{zu}(t,z,u)|<+\infty.
\end{cases}
\end{align}
According to the definition~\eqref{eq:OT} and \eqref{eq:xi-bound}, the region ${\cal O}_T$ has a boundary that is at least $C^1$. We next consider the original HJB equation \eqref{eq:HJBA}, however, restricted to the domain $(t,y,z)\in{\cal O}_T$ that
\begin{align}\label{eq:primehjbOT}
\begin{cases}
\displaystyle v_t+\sup_{\theta\in\R^n} \left[v_x\theta^{\top}\mu+\frac{v_{xx}}{2}\theta^{\top}\sigma\sigma^{\top}\theta+v_{xz}\sigma_Z(z)\theta^{\top}\sigma\gamma\right]\\
\displaystyle\qquad+v_z\mu_Z(z)+v_{zz}\frac{\sigma_Z^2(z)}{2}-f(t,z)v_x=\rho v,\quad (t,z,x)\in {\cal O}_T; \\ \\
\displaystyle v_x(t,z,0)=1,\quad \forall~(t,z)\in[0,T)\times\R.
\end{cases}
\end{align}
First of all, for $(t,z,x)\in{\cal O}_T$, let us define $y^*=y^*(t,z,x)\in(0,1]$ that satisfies
\begin{align}\label{eq:def-xstar}
  \hat{v}_y(t,z,y^*)=-x.
\end{align}
Thanks to \eqref{eq:def-xstar}, we have that
\begin{align}\label{eq:vOT0}
  v(t,z,x)=\inf_{y\in(0,1]}\{\hat{v}(t,z,y)+xy\}=\hat{v}(t,z,y^*(t,z,x))+xy^*(t,z,x),\ \ \ (t,z,x)\in{\cal O}_T.
\end{align}
Note that $(0,1]\ni y\to\hat{v}_y(t,z,y)$ is strictly increasing for fixed $(t,z)\in[0,T]\times\R$, as well as $\hat{v}_y(t,z,1)=0$ and $\lim_{y\to0}\hat{v}_y(t,z,y)=-\xi(t,z)$, we have that $x\to y^*(t,z,x)$ is decreasing, $\lim_{x\to0}y^*(t,z,x)=1$ as well as $\lim_{x\to\xi(t,z)}y^*(t,z,x)=0$. It follows from the implicit function theorem that $y^*$ is $C^1$ on ${\cal O}_T$. Therefore $v$ in \eqref{eq:vOT0} is well defined, and it is $C^{1,2,2}$ on ${\cal O}_T$. On the other hand, a direct calculation yields that, for $(t,z,x)\in{\cal O}_T$,
\begin{align}\label{eq:xstarvderi}
  &y^*(t,z,x)=v_x(t,z,x),\quad\ \hat{v}_t(t,z,y^*(t,z,x))=v_t(t,z,x),\quad\ \hat{v}_z\big(t,z,y^*(t,z,x)\big)=v_z(t,z,x),\notag\\
  &\hat{v}_{yy}(t,z,y^*(t,z,x))\!\!=\!\!-\frac1{v_{xx}(t,z,x)},~\ \ \ \
  \hat{v}_{zy}(t,z,y^*(t,z,x))\!\!=\!\!\frac{v_{xz}(t,z,x)}{v_{xx}(t,z,x)},\notag\\
  &\hat{v}_{zz}(t,z,y^*(t,z,x))\!\!=\!\!\left(v_{zz}-\frac{v^2_{xz}}{v_{xx}}\right)(t,z,x).
\end{align}
Recall that $v_{xx}(t,z,x)<0$ for $(t,z,x)\in{\cal O}_T$. Plugging \eqref{eq:xstarvderi} into \eqref{eq:hatvHJB1}, we deduce that $v$ defined in \eqref{eq:vOT0} solves the dual PDE \eqref{eq:primehjbOT}. We next study the behavior of $v$ on ${\cal O}^c_T\cap([0,T)\times\R\times\R_+)$. To this end, for $(t,z)\in[0,T)\times\R$, let $(t_n,z_n,x_n)\in{\cal O}_T$ for $n\geq1$ be a sequence such that $(t_n,z_n,x_n)\to(t,z,\xi(t,z))$. We claim that
\begin{align}\label{eq:xstar-limit}
  \lim_{n\to+\infty}y^*(t_n,z_n,x_n)=0.
\end{align}

We prove \eqref{eq:xstar-limit} by contradiction. Suppose that, up to a subsequence, there exists a constant $\delta>0$ such that $\lim_{n\to+\infty}y^*(t_n,z_n,x_n)=\delta$. By \eqref{eq:def-xstar}, it yields that
\begin{align*}
  \hat{v}_y(t,z,\delta)=\lim_{n\to+\infty}\hat{v}_y(t_n,z_n,y^*(t_n,z_n,x_n))=-\lim_{n\to+\infty}x_n=-\xi(t,z),
\end{align*}
which contradicts the definition \eqref{eq:xi-bound} of $\xi(t,z)$ and the fact that $\lim_{y\rightarrow 0}\hat{v}_y(t,z,y)=-\xi(t,z)$. Moreover, it follows from \eqref{eq:xstar-limit} that
\begin{align}\label{eq:partial-U-continuity}
  \lim_{n\to+\infty}v(t_n,z_n,x_n)=\lim_{n\to+\infty}\{\hat{v}(t_n,z_n,y^*(t_n,z_n,x_n))+x_ny^*(t_n,z_n,x_n)\}=0.
\end{align}
Similar to the proof of \eqref{eq:hatderivatives}, by \eqref{eq:xstar-limit} again, we also have that
\begin{align}\label{eq:vpartiallimits}
  \lim_{n\to+\infty}v_t(t_n,z_n,x_n)&=\lim_{n\to+\infty}\hat{v}_t(t_n,z_n,y^*(t_n,z_n,x_n))=0,\nonumber\\
  \lim_{n\to+\infty}v_z(t_n,z_n,x_n)&=\lim_{n\to+\infty}\hat{v}_z(t_n,z_n,y^*(t_n,z_n,x_n))=0,\nonumber\\
  \lim_{n\to+\infty}v_x(t_n,z_n,x_n)&=\lim_{n\to+\infty}y^*(t_n,z_n,x_n)=0,\\
  \lim_{n\to+\infty}v_{xx}(t_n,z_n,x_n)&=-\lim_{n\to+\infty}\hat{v}_{yy}(t_n,z_n,y^*(t_n,z_n,x_n))^{-1}=0,\nonumber\\
  \lim_{n\to+\infty}v_{xz}(t_n,z_n,x_n)&=-\lim_{n\to+\infty}\frac{\hat v_{yz}}{\hat v_{yy}}(t_n,z_n,y^*(t_n,z_n,x_n))=0,\nonumber\\
  \lim_{n\to+\infty}v_{zz}(t_n,z_n,x_n)&=\lim_{n\to+\infty}\left(\hat v_{zz}-\frac{\hat v^2_{yz}}{\hat v_{yy}}\right)(t_n,z_n,y^*(t_n,z_n,x_n))=0.\nonumber
\end{align}
Let us define $v(t,z,x)=0$ for $(t,z,x)\in {\cal O}^c_T\cap([0,T)\times\R\times\R_+)$. By \eqref{eq:partial-U-continuity} and \eqref{eq:vpartiallimits},  we have that $v$ given by \eqref{eq:vOT0} and its partial derivatives up to order two are continuous on $\partial{\cal O}_T\cap([0,T)\times\R\times\R_+)$. Therefore, $v$ is $C^{1,2,2}$ on $[0,T)\times\R\times\R_+$. Moreover, using \eqref{eq:xstarvderi} and \eqref{eq:hatvHJB1} on $[0,T]\times\R\times\R_+$, we have that $v$ given by \eqref{eq:solution-prime-hjb} solves the following HJB equation:
\begin{align}\label{eq:value-HJB1}
\begin{cases}
\displaystyle v_t+\sup_{\theta\in\R^n} \left[v_x\theta^{\top}\mu+\frac{v_{xx}}{2}\theta^{\top}\sigma\sigma^{\top}\theta+v_{xz}\sigma_Z(z)\theta^{\top}\sigma\gamma\right]\\
\displaystyle\qquad+v_z\mu_Z(z)+v_{zz}\frac{\sigma_Z^2(z)}{2}-f(t,z)v_x=\rho v,\quad \forall~(t,z,x)\in[0,T)\times\R\times\R_+; \\ \\
\displaystyle v_x(t,z,0)=1,\quad \forall~(t,z)\in[0,T)\times\R.
\end{cases}
\end{align}
On the other hand, note that $v_x\geq0$ on $(t,z,x)\in[0,T)\times\R\times\R_+$, and $v(t,z,x)\to 0$ as $x\to+\infty$. By \eqref{eq:valuefcnA}, it is easy to see $v(t,z,0)\leq0$. Therefore, $v(t,z,x)\leq0$ for $(t,z,x)\in[0,T)\times\R\times[0,\infty)$ and it follows from \eqref{eq:estidualtransfromh} that there exists a constant $C>0$ independent of $T$ that
\begin{align}\label{eq:vOT}
  \left|v(t,z,x)\right|&=-v(t,z,x)=\sup_{y\in(0,1]}\{-\hat{v}(t,z,y)-xy\}\leq\sup_{y\in(0,1]}\{-\hat{v}(t,z,y)\}\notag\\
  &=\sup_{y\in(0,1]}\{-e^{\rho t}h(t,z,-\ln y)\}\leq e^{\rho t}C(T-t)(1+|z|),
\end{align}
for $(t,z,x)\in[0,T)\times\R\times[0,\infty)$, where the function $h(t,z,u)$ is given by \eqref{eq:dualtransfromh}.

We next prove the continuity of $v$ on the boundary of $[0,T)\times\R\times[0,+\infty)$. Note that $v(t,z,0)=\hat{v}(t,z,1)$ and let us consider $(t,z)\in[0,T)\times\R$ and $(t_n,z_n,x_n)\in[0,T)\times\R\times\R_+$ satisfying $(t_n,z_n,x_n)\to(t,z,0)$ as $n\to\infty$. By mimicking the proof showing \eqref{eq:xstar-limit}, one can obtain that
\begin{align}\label{x^*-continuity-1}
  \lim_{n\to\infty}y^*(t_n,z_n,x_n)=1.
\end{align}
An application of L'Hospital's rule gives that
\begin{align*}
  \lim_{x\downarrow0}\frac1x\left(v(t,z,x)-v(t,z,0)\right)&=\lim_{x\downarrow0}\frac{1}{x}\big(\hat{v}(t,z,y^*(t,z,x))+xy^*(t,z,x)-\hat{v}(t,z,1)\big)
  \notag\\
  &=\lim_{x\downarrow0}y^*(t,z,x)-\lim_{x\downarrow0}\frac{\hat{v}_y(t,z,y^*(t,z,x))-\hat{v}(t,z,1)}{y^*(t,z,x)-1}\times
  \lim_{x\downarrow0}\frac{y^*(t,z,x)-1}x\notag\\
  &=1-\hat{v}_y(t,z,1)\times\left(\lim_{x\downarrow0}y^*_x(t,z,x)\right)=1.
\end{align*}
Moreover, as $\lim_{n\to\infty}v_x(t_n,z_n,x_n)=\lim_{n\to\infty}y^*(t_n,z_n,x_n)=1$, it holds that
\begin{align}\label{eq:continuos-vz}
  \lim_{n\to\infty}v_x(t_n,z_n,x_n)=v_x(t,z,0).
\end{align}
Similarly, we also have that
\begin{align*}
  \lim_{x\downarrow0}\frac1x\left(v_x(t,z,x)-v_x(t,z,0)\right)=\lim_{x\downarrow0}\frac{1}{x}\left(y^*(t,z,x)-1\right)
  =\lim_{x\downarrow0}y^*_x(t,z,x)=-\hat{v}_{yy}(t,z,1)^{-1},
\end{align*}
and $\lim_{n\to\infty}v_{xx}(t_n,z_n,x_n)=-\lim_{n\to+\infty}\hat{v}_{yy}(t_n,z_n,y^*(t_n,z_n,x_n))^{-1}=-\hat{v}_{yy}(t,z,1)^{-1}$.
Therefore
\begin{align}\label{eq:continuos-vzz}
  \lim_{n\rightarrow+\infty}v_{xx}(t_n,z_n,x_n)=v_{xx}(t,z,0).
\end{align}
In a similar fashion, the limits \eqref{eq:continuos-vz} and \eqref{eq:continuos-vzz} also hold for $v_z$, $v_{xz}$, and $v_{zz}$. Hence, we conclude that $v\in C^{1,2,2}([0,T)\times\R\times[0,\infty))$.

On the other hand, for $(z,x)\in\R\times[0,+\infty)$, we define $v(T,z,x)=0$ and consider $(t_n,z_n,x_n)\in[0,T)\times\R\times[0,+\infty)$ satisfying $(t_n,z_n,x_n)\to(T,z,x)$ as $n\to+\infty$. In view of \eqref{eq:vOT}, we have $\lim_{n\to+\infty}v(t_n,z_n,x_n)=0$, which yields that $v\in C({\cal D}_T)$. By combining Eq.~\eqref{eq:value-HJB1}, we deduce that $v\in C^{1,2,2}([0,T)\times\R\times[0,+\infty))\cap C({\cal D}_T)$ and $v$ satisfies that
\begin{align}\label{eq:value-HJB}
\begin{cases}
\displaystyle v_t+\sup_{\theta\in\R^n} \left[v_x\theta^{\top}\mu+\frac{v_{xx}}{2}\theta^{\top}\sigma\sigma^{\top}\theta+v_{xz}\sigma_Z(z)\theta^{\top}\sigma\gamma\right]\\
\displaystyle\qquad+v_z\mu_Z(z)+v_{zz}\frac{\sigma_Z^2(z)}{2}-f(t,z)v_x=\rho v,\quad \forall~(t,z,x)\in[0,T)\times\R\times\R_+; \\ \\
\displaystyle v_x(t,z,0)=1,\quad \forall~(t,z)\in[0,T)\times\R,\\ \\
\displaystyle v(T,z,x)=0,\quad\forall~(z,x)\in\R\times[0,+\infty),
\end{cases}
\end{align}
and the estimate \eqref{eq:vOT} holds for $(t,z,x)\in[0,T]\times\R\times[0,+\infty)$.

We next prove {\bf(ii)}. We first show the continuity of $\theta^*(t,z,x)$ on $(t,z,x)\in\mathcal{D}_T$, which verifies the admissibility of $\theta_t^*=\theta^*(t,Z_t,X_t)$ for $t\in[0,T]$ (i.e., $\theta^*\in{\cal U}_t$). Let us define $y^*(t,z,0)=1$. Thanks to \eqref{x^*-continuity-1}, $y^*$ is continuous at $(t,z,0)$. For $(t,z,x)\in\mathcal{D}_T$, we rewrite \eqref{eq:verifioptimalstrategy} by
\begin{align}\label{eq:alternative-thetastar}
\theta^*(t,z,x)
&\!\!=\!\!-\mu(\sigma\sigma^\top)^{-1}y^*(t,z,x)\hat{v}_{yy}\left(t,z,y^*(t,z,x)\right)\!\!+\!(\sigma\sigma^\top)^{-1}\sigma_Z(z)\sigma\gamma\hat{v}_{yz}\left(t,z,y^*(t,z,x)\right).
\end{align}
It is easy to see that $\theta^*(t,z,x)$ is continuous for $(t,z,x)\in{\cal O}_T\cup\{(t,z,0);~(t,z)\in[0,T)\times\R\}$. Therefore, it remains to show that
\begin{align}\label{continuous-theta*}
  \lim_{n\rightarrow+\infty}\theta^*(t_n,z_n,x_n)=\theta^*(t,z,x),
\end{align}
where $x=\xi(t,z)$, and $(t_n,z_n,x_n)\in{\cal O}_T$, $\lim_{n\rightarrow+\infty}(t_n,z_n,x_n)=(t,z,x)$. By virtue of \eqref{eq:alternative-thetastar}, we have that
\begin{align}\label{eq:alternative-thetastar1}
\theta^*(t_n,z_n,x_n)
&=-\mu(\sigma\sigma^\top)^{-1}y^*(t_n,z_n,x_n)\hat{v}_{yy}(t,z,y^*(t_n,z_n,x_n))\nonumber\\
&\quad+(\sigma\sigma^\top)^{-1}\sigma_Z(z)\sigma\gamma\hat{v}_{yz}(t,z,y^*(t_n,z_n,x_n)).
\end{align}
Note that $x^*(t_n,z_n,x_n)\rightarrow0$ as $n\rightarrow+\infty$. By sending $n$ to $+\infty$ on both sides of \eqref{eq:alternative-thetastar1}, we deduce that
\begin{align*}
  \lim_{n\to\infty}\theta^*(t_n,z_n,x_n)&=-\mu(\sigma\sigma^\top)^{-1}\lim_{y\downarrow0}x\hat{v}_{yy}(t,z,y)
  +(\sigma\sigma^\top)^{-1}\sigma_Z(z)\sigma\gamma\lim_{y\downarrow0}\hat{v}_{yz}(t,z,y)
  =\theta^*(t,z,x).
\end{align*}
Following the same argument, we can establish the convergence \eqref{continuous-theta*} for any $(t,z,x)\in\partial{\cal O}_T\cap([0,T]\times\R\times[0,+\infty))$ and $(t_n,z_n,x_n)\in\mathcal{O}_T$. Hence $\theta^*(t,y,z)$ is continuous for $(t,z,x)\in[0,T]\times\R\times[0,+\infty)$. Moreover, one can see from \eqref{hat-v}, \eqref{eq:verifioptimalstrategy} and \eqref{eq:alternative-thetastar} that there exists constant $C>0$ such that
\begin{align}\label{eq:theta^*-estimate}
  |\theta^*(t,z,x)|\leq C(1+|z|),\quad\forall~(t,z,x)\in\mathcal{D}_T.
\end{align}
With the continuity of $\theta^*$ on $\mathcal{D}_T$ and the estimate \eqref{eq:theta^*-estimate}, we can apply Theorem $2.2$, Theorem $2.4$ and Remark $2.1$ in chapter 4 of \cite{IkedaW1992} to conclude that the following SDE admits a weak solution that: for $t\in[0,T]$,
\begin{align}\label{eq:weak-solution}
\begin{cases}
\displaystyle \tilde{X}_s = -\int_t^s f(r,Z_r)dr+ \int_t^s\theta^*(r,Z_r,\Phi^t(\tilde{X})_r)^{\top}\mu dr+\int_t^s\theta^*(r,Z_r,\Phi^t(\tilde{X})_r)^\top\sigma dW_r,\\ \\
\displaystyle dZ_s =\mu_Z(Z_s) ds + \sigma_{Z}(Z_s)dW^{\gamma}_s,\ \ s\in[t,T].
\end{cases}
\end{align}
Here, the mapping $\Phi^t:C([t,T];\R)\to C([t,T];\R)$ satisfies that, for all $\varphi\in C([t,T];\R)$,
\begin{itemize}
  \item[(i)]$\Phi^t(\varphi)_s=\varphi_s+\eta_s$ for $s\in[t,T]$, and $\Phi(\varphi)_t=\varphi_t$.
  \item[(ii)]$\Phi^t(\varphi)_s\geq0$ for $s\in[t,T]$.
  \item[(iii)] $s\to\eta_s$ is continuous, non-negative and non-decreasing, and $\eta_s=x\vee\int_t^s{\bf1}_{\{\Phi^t(\varphi)_r=0\}}d\eta_r$ for $s\in[t,T]$.
\end{itemize}
Define $X^*:=\Phi^t(\tilde{X})$ and $L^{*}:=\Phi^t(\tilde{X})-\tilde{X}$. Then $(X^*,L^*,W)$ solves the SDE that
\begin{align*}
X_s^* = -\int_t^s f(r,Z_r)dr+ \int_t^s\theta^*(r,Z_r,X_r^*)^{\top}\mu dr+\int_t^s\theta^*(r,Z_r,X_r^*)^{\top}\sigma dW_r + L_s^*,\ \ s\in[t,T],
\end{align*}
where $L^*$ satisfies (iii). This shows that $\theta^*\in{\cal U}_t$ is admissible.

Let us fix any $(t,z,x)\in[0,T)\times\R\times[0,\infty)$, and $\theta\in{\cal U}_t$. For any $n>T^{-1}$, we define that
\begin{align}
  \tau_n^t:=\left(T-\frac{1}{n}\right)\wedge\inf\left\{s\geq t:\left|Z_s\right|+\left|X_s\right|> n\right\}.
\end{align}
It holds that $\tau_n^t\uparrow T$ as $n\to\infty$, $\Px$-a.s.. By It\^{o}'s formula, we get that
\begin{align}\label{eq:verification}
&\Ex_{t,z,x}\left[-\int_t^{\tau_n^t}e^{-\rho s}dL_s^X+e^{-\rho {\tau_n^t}}v(\tau_n^t,Z_{\tau_n^t},X_{\tau_n^t})-e^{-\rho t}v(t,Z_t,X_t)\right]\\
&=\Ex_{t,z,x}\left[\int_t^{\tau_n^t}e^{-\rho s}\left(v_t+{\cal L}^{\theta_s} v\right)(s,Z_s,X_s)ds\right]+\Ex_{t,z,x}\left[\int_t^{\tau_n^t}e^{-\rho s}\left(v_z(s,Z_s,X_s)-1\right)dL_s^X\right],\notag
\end{align}
where, for $\theta\in\R^n$, the operator ${\cal L}_t^{\theta}$ acted on $C^{2}(\R\times[0,\infty))$ is defined by
\begin{align*}
{\cal L}_t^\theta\varphi(z,x)&:=\varphi_x(z,x)\theta^{\top}\mu+\frac{\varphi_{xx}(z,x)}{2}\theta^{\top}\sigma\sigma^{\top}\theta
+\varphi_{xz}(z,x)\sigma_Z(z)\theta^{\top}\sigma\gamma+\varphi_z(z,x)\mu_Z(z)\nonumber\\
&\quad+\varphi_{zz}(z,x)\frac{\sigma_Z^2(z)}{2}-f(t,z)\varphi_x(z,x)-\rho \varphi(z,x),
\end{align*}
for all $\varphi\in C^{2}(\R\times[0,\infty))$. Then, \eqref{eq:reflectedZA} and the boundary condition in \eqref{eq:value-HJB} yield that
\begin{align*}
&\Ex_{t,z,x}\left[\int_t^{\tau_n^t}e^{-\rho s}\left(v_x(s,Z_s,X_s)-1\right)dL_s^X\right]=\Ex_{t,z,x}\left[\int_t^{\tau_n^t}e^{-\rho s}\left(v_x(s,Z_s,X_s)-1\right)\mathbf{1}_{\left\{X_s=0\right\}}dL_s^X\right]=0.
\end{align*}
On the other hand, the HJB equation \eqref{eq:value-HJB} satisfied by $v$ also gives that, for all $t\in[0,T]$, $\Px$-a.s.
\begin{align}\label{eq:HJB-variance}
  \left(v_t+\mathcal{L}^\theta_t v\right)(t,Z_t,X_t)\leq0,
\end{align}
where the equality holds in \eqref{eq:HJB-variance} if $\theta=\theta^*$. 
Hence, we deduce from \eqref{eq:HJB-variance} that
\begin{align}\label{eq:verification1}
  \Ex_{t,z,x}\left[-\int_t^{\tau_n^t}e^{-\rho s}dL_s^X\right]\leq e^{-\rho t}v(t,z,x)-\Ex_{t,z,x}\left[e^{-\rho\tau_n^t}v(\tau_n^t,Z_{\tau_n^t},X_{\tau_n^t})\right].
\end{align}

By applying the estimate \eqref{eq:vOT}, we have $|v(\tau_n^t,Z_{\tau_n^t},X_{\tau_n^t})|\leq e^{\rho\tau_n^t}C(T-\tau_n^t)\{1+\sup_{s\in[t,T]}|Z_s|^p\}$, $\Px$-a.s.
Sending $n\rightarrow+\infty$ and noting that $\tau_n^t\uparrow T$, $\Px$-a.s., we can see from dominated convergence theorem that
\begin{align}\label{eq:terminalTcond}
  \lim_{n\to+\infty}\Ex_{t,z,x}\left[e^{-\rho {\tau_n}}v(\tau_n^t,Z_{\tau_n^t},X_{\tau_n^t})\right]=0.
\end{align}
Therefore, as $n\rightarrow+\infty$ in \eqref{eq:verification1}, we have that, for all $\theta\in{\cal U}_t$,
\begin{align}\label{eq:verification2}
J(\theta;t,z,x)=\Ex_{t,z,x}\left[-\int_t^Te^{-\rho s}dL_s^X\right]\leq e^{-\rho t}v(t,z,x)={\rm w}(t,z,x),\quad (t,z,x)\in\mathcal{D}_T,
\end{align}
where the equality in \eqref{eq:verification2} holds for $\theta=\theta^*$. This finally verifies that $\theta^*\in{\cal U}_t$ is an optimal strategy.
\end{proof}

\section{Illustrative Examples}\label{sec:exam}
Although the benchmark process is restricted to be non-decreasing in \eqref{eq_prob_IBP}, this section illustrates the application to some market index tracking problems when the  index process follows a geometric Brownian motion. In both cases of finite horizon and infinite horizon, we show that the market index tracking problem can actually be transformed into an equivalent optimal tracking problem with a non-decreasing benchmark process. Some closed-form results can be derived therein, allowing us to numerically examine the sensitivity on model parameters and discuss some financial implications.

Similar to the model assumption in (22) of \cite{GHW2011}, we consider the market index such as S\&P500 or Nasdaq 100 with the price process $I=(I_t)_{t\in[0,T]}$ that satisfies
\begin{align}\label{eq:sDtsp}
\frac{dI_t}{I_t}=\mu_Idt + \sigma_IdW_t^{\gamma},
\end{align}
where $I_0=z>0$ and the constant return rate $\mu_I\in\R$ and the volatility $\sigma_I>0$. Note that the index process in \eqref{eq:sDtsp} is not monotone and hence does not fit directly into our framework in Section \ref{sec:model}. For the given market index process $I$, we consider the optimal tracking problem similar to \eqref{eq_prob_IBP} that
\begin{align}\label{eq:probIBPSsp}
u(a, \mathrm{v}, z):=\text{$\inf_{C,\theta} \Ex\left[C_0+\int_0^T e^{-\rho t}dC_t \right]$\ \ \ subject to\ \  $I_t\le V_t^{\theta} + C_t$ at each $t\in[0,T]$.}
\end{align}
To exclude the trivial case, it is assumed that $I$ can not be dynamically replicated by some portfolio $\theta$, which amounts to the assumption that
\begin{align}
\lambda:=\mu_I-\sigma_I\gamma^{\top}\sigma^{-1}\mu\neq0. \label{lambda}
\end{align}

It is not difficult to see that Lemma~\ref{first-transf} still holds when $A_t$ is replaced by the market index process $I_t$. We can therefore follow the argument in Section~\ref{sec:PDE} and introduce the reflected state process that
\begin{align}\label{state-X-example}
X_t &= -(I_t-I_0) +\int_0^t\theta_s^{\top}\mu ds+\int_0^t\theta_s^{\top}\sigma dW_s + L_t\notag\\
&=-\int_0^tf(I_s)ds+\int_0^t\bar{\theta}_s^{\top}\mu ds+\int_0^t\bar{\theta}_s^{\top}\sigma dW_s+L_t,
\end{align}
where the running maximum process $L=(L_t)_{t\in[0,T]}$  with $L_0=x\geq0$ is defined as in \eqref{eq:maxM} when $A$ is replaced by $I$. In \eqref{state-X-example}, we used the notations $f(z):=\lambda z$ for $z>0$, and for $t\in[0,T]$,
\begin{align}\label{eq:bartheta}
\bar{\theta}_t^{\top}:= \theta_t^{\top}-\sigma_I\gamma^{\top}\sigma^{-1}I_t.
\end{align}
It follows that $(a-\mathrm{v})^+-u(a,\mathrm{v}, z)$ with $u(a,\mathrm{v},z)$ given in \eqref{eq:probIBPSsp} is equivalent to the auxiliary control problem
\begin{align}\label{eq:objectivefcnGBM}
\sup_{\bar{\theta}\in\overline{\mathcal{U}}_0}\Ex\left[-\int_0^{T}e^{-\rho s}dL_s\right],
\end{align}
where the initial level $L_0^X=x=(\mathrm{v}-a)^+$.

Let $\mathcal{D}_T:=[0,T]\times(0,\infty)\times[0,\infty)$. Denote by $\overline{\cal U}_t$ the set of admissible controls taking the feedback form $\bar{\theta}_s=\bar{\theta}(s,I_s,X_s)$ for $s\in[t,T]$, where $\bar{\theta}:\mathcal{D}_T\to\R^n$ is a measurable function such that the following reflected SDE has a weak solution:
\begin{align}\label{eq:reflectedZAGBM}
X_s = -\int_t^s f(I_r)dr+ \int_t^s\bar{\theta}(r,I_r,X_r)^{\top}\mu dr+\int_0^t\bar{\theta}(r,I_r,X_r)^{\top}\sigma dW_r + L_s^{X},
\end{align}
with $X_t=x\geq 0$. Here, $L_s^X=x\vee \int_t^s{\bf1}_{\{X_r=0\}}dL_r^X$ is a continuous, non-negative and non-decreasing process, which increases only when the state process $X_s$ hits the level $0$ for $s\in[t,T]$. For $(t,z,x)\in\mathcal{D}_T$, the dynamic version of the auxiliary problem \eqref{eq:objectivefcnGBM} is given by
\begin{align}\label{eq:valuefcnAGBM}
{\rm w}(t,z,x) :=\sup_{\bar{\theta}\in\overline{\cal U}_t}\Ex_{t,z,x}\left[-\int_t^Te^{-\rho s}dL_s^X\right],
\end{align}
where $\Ex_{t,z,x}[~\cdot~]:=\Ex[~\cdot~|I_t=z,X_t=x]$. Again, we shall consider $v(t,z,x):=e^{\rho t}{\rm w}(t,z,x)$ as the solution to the primal HJB equation.

If the coefficient $\lambda$ in \eqref{lambda} satisfies $\lambda<0$, the optimal strategy is actually trivial that $\bar{\theta}^*\equiv0$, and the associated value function becomes $v(t,z,x)=0$ for all $(t,z,x)\in[0,T]\times\R_+^2$. In fact, let $X^*=(X_s^*)_{s\in[t,T]}$ be the state process \eqref{eq:reflectedZAGBM} with $\bar{\theta}^*=(\bar{\theta}_s^*)_{s\in[t,T]}\in{\overline{\cal U}}_t$. Then, for $\lambda<0$, it holds that $X^*_s>0$ for all $s\in[t,T]$. This yields that $L^{X^*}_s=L^{X^*}_t=x>0$ for all $s\in[t,T]$, a.s. and hence ${\rm w}(t,z,x)\leq0=\Ex_{t,z,x}\left[-\int_t^Te^{-\rho s}dL_s^{X^*}\right]$.

If, on the other hand, the coefficient $\lambda>0$, i.e., $f(z)>0$, the auxiliary stochastic control problem \eqref{eq:valuefcnAGBM} with the state process \eqref{state-X-example} falls into the framework in Section~\ref{sec:PDE}-Section~\ref{sec:verify} under the assumption that $f(t,x)>0$. Hence, all results still hold and Theorem \ref{thm:verificationthem} gives the characterization of the value function and the optimal portfolio for the index tracking problem \eqref{eq:valuefcnAGBM}. To further explore more explicit results, let us focus on the special case $\sigma_I=0$, i.e., the benchmark process $I_t=ze^{\mu t}$ for $t\in[0,T]$ is the so-called \textit{growth rate benchmark} that has been studied in \cite{YaoZZ06}. With this deterministic benchmark $I=(I_t)_{t\in[0,T]}$, the dual HJB equation \eqref{eq:frenchelHJB1} can be reduced to
\begin{align}\label{eq:dualhjbsigmasp=0}
\hat{v}_t(t,z,y)-\rho\hat{v}(t,z,y)+\rho y\hat{v}_y(t,z,y)+\alpha y^2\hat{v}_{yy}(t,z,y)+\mu_Iz\hat{v}_z(t,z,y)=\lambda zy.
\end{align}
In view of \eqref{hat-v}, it holds that $\hat{v}(t,z,x)=e^{\rho t}h(t,z,-\ln x)$ for $(t,z,x)\in[0,T]\times\R_+^2$. It follows from \eqref{eq:dualtransfromh} that, for $(t,z,u)\in{\cal D}_T$,
\begin{align}\label{eq:hexpsigmasp=0}
h(t,z,u) =-\lambda ze^{-\mu_It}\Ex\left[\int_t^T e^{(\mu_I-\rho)s-R_s^{t,u}}ds\right]=-\lambda ze^{-\mu_It}\int_t^T e^{(\mu_I-\rho)s}\Ex\left[e^{-R_s^{t,u}}\right]ds,
\end{align}
where we recall that $(R_s^{t,u})_{s\in[t,T]}$ with $(t,u)\in[0,T]\times[0,\infty)$ is a reflected Brownian motion with drift defined by \eqref{eq:driftBM}.

We now compute the term $\Ex[e^{-R_s^{t,u}}]$ in \eqref{eq:hexpsigmasp=0}. By \cite{Harrison1985} on page 49, for $(t,m)\in\R_+\times\R$, we have that
\begin{align}\label{eq:distribution-Rou}
\Px\left(R_t^{0,u}\leq m\right) =\Phi\left(\frac{-u+m-(\alpha-\rho)t}{\sqrt{2\alpha t}}\right)-e^{\frac{(\alpha-\rho)m}{\alpha}}\Phi\left(\frac{-u-m-(\alpha-\rho)t}{\sqrt{2\alpha t}}\right),
\end{align}
where $\Phi(m)=\int_{-\infty}^{m}\frac{1}{\sqrt{2\pi}}e^{-\frac{u^2}{2}}du$, $m\in\R$, denotes the standard normal cumulative distribution function. It follows from the Markov property that, for $(s,m)\in(t,T]\times\R_+$,
\begin{align}\label{eq:density-Rou}
&\frac{\Px(R_s^{t,u}\in dm)}{dm}=\frac{1}{\sqrt{2\alpha(s-t)}}\Phi'\left(\frac{-u+m-(\alpha-\rho)(s-t)}{\sqrt{2\alpha (s-t)}}\right)
-\frac{\alpha-\rho}{\alpha}e^{\frac{(\alpha-\rho)m}{\alpha}}\nonumber\\
&\qquad\quad\times\Phi\left(\frac{-u-m-(\alpha-\rho)(s-t)}{\sqrt{2\alpha(s-t)}}\right)
+\frac{1}{\sqrt{2\alpha(s-t)}}e^{\frac{(\alpha-\rho)m}{\alpha}}\Phi'\left(\frac{-u-m-(\alpha-\rho)(s-t)}{\sqrt{2\alpha(s-t)}}\right)\nonumber\\
&\qquad=: \psi(m,u,s-t).
\end{align}
The expectation can therefore be explicitly written as
\begin{align*}
\Ex\left[e^{-R_s^{t,u}}\right]=\int_0^{\infty}e^{-m}\Px(R_s^{t,u}\in dm)=\int_0^{\infty}e^{-m}\psi(m,u,s-t)dm.
\end{align*}
By virtue of \eqref{eq:hexpsigmasp=0}, the dual HJB equation \eqref{eq:dualhjbsigmasp=0} admits the solution in an integral form that
\begin{align}\label{eq:hatvexpsigmasp=02}
\hat{v}(t,z,x)=e^{\rho t}h(t,z,-\ln x)=-\lambda ze^{(\rho-\mu_I)t}\int_t^T\int_{0}^{\infty} e^{(\mu_I-\rho)s-m}\psi(m,-\ln x,s-t)dmds.
\end{align}
Moreover, the critical point defined by \eqref{eq:xi-bound} in Theorem \ref{thm:verificationthem} can also be explicitly computed that
\begin{align*}
  \xi(t,z)&=\lambda z\Ex\left[\int_t^Te^{(\alpha+\mu_I-2\rho)(s-t)}e^{\sqrt{2\alpha}B^1_{s-t}}ds\right]=\lambda z\int_t^T e^{(2\alpha+\mu_I-2\rho)(s-t)} ds\\
  &=\begin{cases}
  \frac{\lambda z}{2\alpha+\mu_I-2\rho}\left[e^{(2\alpha+\mu_I-2\rho)(T-t)}-1\right], & 2\alpha+\mu_I-2\rho\neq0;\\ \\
  \qquad\lambda z(T-t), & 2\alpha+\mu_I-2\rho=0.
  \end{cases}
\end{align*}

Next, we shall consider the same optimal tracking problem \eqref{eq:probIBPSsp} with the infinite time horizon (i.e., $T\to\infty$), in which the value function and optimal portfolio can be obtained explicitly. For the infinite horizon control problem, we note that ${\rm w}(z,x)=v(z,x)$ is the solution to the stationary version of the HJB equation \eqref{eq:HJBA1} that
\begin{align}\label{eq:stationaryHJBGBM}
-\rho v-\alpha \frac{v_x^2}{v_{xx}}+\frac{\sigma_I^2}{2}z^2\left(v_{zz}-\frac{v_{xz}^2}{v_{xx}}\right)-\sigma_I\mu^{\top}(\sigma\sigma^{\top})^{-1}\sigma\gamma z\frac{v_xv_{xz}}{v_{xx}}+
\mu_Izv_z-\lambda zv_x=0,
\end{align}
with the Neumann boundary condition $v_x(z,0)=1$. The dual PDE for \eqref{eq:stationaryHJBGBM} becomes
\begin{align}\label{eq:dualstationaryHJBGBM}
  &-\rho\hat{v}(z,y)+\rho y\hat{v}_y(z,y)+\alpha y^2\hat{v}_{yy}(z,y)+\mu_Iz\hat{v}_z(z,y)+\frac{\sigma_I^2}{2}z^2\hat{v}_{zz}(z,y)\nonumber\\
&\qquad-\sigma_I\mu^{\top}(\sigma\sigma^{\top})^{-1}\sigma\gamma yz\hat{v}_{yz}(z,y)-\lambda zy=0,\quad (z,y)\in\R_+\times(0,1).
\end{align}
In view of \eqref{eq:Neumann-boundary-hatv}, we have that $\hat{v}_y(z,1)=0$. It is easy to verify that the dual equation \eqref{eq:dualstationaryHJBGBM} admits the following general solution given by
\begin{align*}
\hat{v}(z,y) =\lambda zh(y)=\lambda z\left\{\frac{1}{\lambda}y+C_1y^{\gamma_1}+C_2y^{\gamma_2}\right\},\quad (z,y)\in\R_+\times(0,1].
\end{align*}
where $C_1,C_2\in\R$ are unknown constants that can be determined later and $\gamma_i$, $i=1,2$, satisfy
\begin{align*}
  (\mu_I-\rho)+(\rho-\sigma_I\gamma^\top\sigma^{-1}\mu)\gamma_i+\alpha\gamma_i(\gamma_i-1)=0,\quad i=1,2.
\end{align*}

Let us assume that the discount factor $\rho>\mu_I$. It holds that $-\infty<\gamma_1<0<\gamma_2<1$. Note that it is required that $h(0)\leq0$ and $h$ is convex. By $h(1)=\frac{1}{\lambda}+C_1+C_2$, we must have that $C_1=0$. It then follows from $\hat{v}_y(z,1)=0$ that $C_2=-\frac1{\gamma_2\lambda}$. Therefore, we conclude that
\begin{align}\label{eq:genealdualGBM}
\hat{v}(z,y) = yz -\frac{1}{\gamma_2}y^{\gamma_2}z,\quad (z,y)\in\R_+\times(0,1].
\end{align}
Moreover, as $T\to\infty$, it is easy to verify that the counterpart of $\xi(t,z)$ in \eqref{eq:xi-bound} is given by
\begin{align*}
  \xi(z):=\Ex\left[\int_0^{\infty}e^{-\rho s}f(M_s^{0,z})e^{\sqrt{2\alpha}B^1_s+(\rho-\alpha)s}ds\right]=z\int_0^{+\infty}e^{\lambda s}ds=+\infty.
\end{align*}
Using the same argument to derive \eqref{eq:solution-prime-hjb} with the closed-form solution \eqref{eq:genealdualGBM}, the solution of the primal HJB equation \eqref{eq:stationaryHJBGBM} can be obtained by
\begin{align*}
  v(z,x)&=\inf_{y\in(0,1]}\{\hat{v}(z,y)+xy\}=\hat{v}(z,y^*(z,x))+xy^*(z,x),\ \ (z,x)\in\R_+^2,
 \end{align*}
where $y^*(z,x)=(1+\frac xz)^{\frac1{\gamma_2-1}}<1$ and hence $y^*(z,x)^{\gamma_2-1}=1+\frac{x}{z}$. A straightforward calculation yields that, for $(z,x)\in\R_+^2$,
\begin{align}\label{eq:clode-formsolHJBGBM}
v(z,x)&=z\left(1+\frac{x}{z}\right)^{\frac{1}{\gamma_2-1}}\left(1-\frac{1}{\gamma_2}-\frac{1}{\gamma_2}\frac{x}{z}\right)
+x\left(1+\frac{x}{z}\right)^{\frac{1}{\gamma_2-1}}=z\frac{\gamma_2-1}{\gamma_2}\left(1+\frac xz\right)^{\frac{\gamma_2}{\gamma_2-1}}<0.
\end{align}
It follows that $v_x(z,x)=(1+\frac{x}{z})^{\frac1{\gamma_2-1}}$ and $v_x(z,0)=1$ for all $z\in\R_+$, i.e., the Neumann boundary condition holds. Moreover, by the similar argument to derive \eqref{eq:verifioptimalstrategy}, the feedback optimal portfolio is obtained by
\begin{align}\label{eq:barthetastar}
 \bar{\theta}^*(z,x)&= -(\sigma\sigma^{\top})^{-1}\frac{v_x(z,x)\mu+zv_{xz}(z,x)\sigma_I\sigma\gamma}{v_{xx}(z,x)}\notag\\
 &=-(\gamma_2-1)(x+z)(\sigma\sigma^{\top})^{-1}\mu+(\gamma_2-1)\sigma_I\left(\frac{z^3}x+z^2\right)\sigma\gamma.
\end{align}

We can show the optimality of $\bar{\theta}^*$ by the verification argument similar to Theorem \ref{thm:verificationthem}. In particular, as we have $\rho>\mu_I$, the counterpart of \eqref{eq:terminalTcond} for the infinite horizon case can be checked by
\begin{align*}
  &\lim_{s\to+\infty}\Ex\left[e^{-\rho s}\big|v(I_s,X_s)\big|\right]=\frac{1-\gamma_2}{\gamma_2}\lim_{s\to+\infty}\Ex\left[e^{-\rho s}I_s\left(1+\frac {X_s}{I_s}\right)^{\frac{\gamma_2}{\gamma_2-1}}\right]\notag\\
  &\qquad\leq\frac{1-\gamma_2}{\gamma_2}\lim_{s\to+\infty}\Ex\left[e^{-\rho s}I_s\right]=\frac{1-\gamma_2}{\gamma_2}z\lim_{s\to+\infty}e^{-(\rho-\mu_I)s}=0.
\end{align*}

Based on \eqref{eq:bartheta}, \eqref{eq:clode-formsolHJBGBM} and \eqref{eq:barthetastar}, we can readily perform some sensitivity analysis on model parameters. Let us consider the case $d=1$ and fix the variable $z=1$. In the following figures, we plot the value function $v(1,x)$ and the feedback optimal portfolio $\theta^*(1,x)$ in terms of the variable $x$ by varying model parameters. In Figure 1, the left panel shows that the value function $v$ is decreasing in the index return $\mu_I$, which indicates that the fund manager needs to inject more total capital when the index performs well. The right panel illustrates that there exists a critical point for the wealth $x$ such that the monotonicity of the optimal portfolio $\theta^*$ in the parameter $\mu_I$ overturns. When the wealth level is too low, the optimal portfolio is decreasing in $\mu_I$, which is consistent with the real life situation that the fund manager with inadequate wealth will become more conservative in the risky investment because any account loss may require a large amount of capital injection to catch up with the high index growth. On the other hand, when the wealth is sufficient, the optimal portfolio is increasing in $\mu_I$, which indicates that the fund manager will become more aggressive to invest in order to compete with the index performance. However, as $\mu_I$ is much larger than $\mu$, the fund manager still needs to increase the injected capital to meet the floor constraint.

$$
\begin{array}{ccc}
\begin{array}{c}
\includegraphics[height=2.1in]{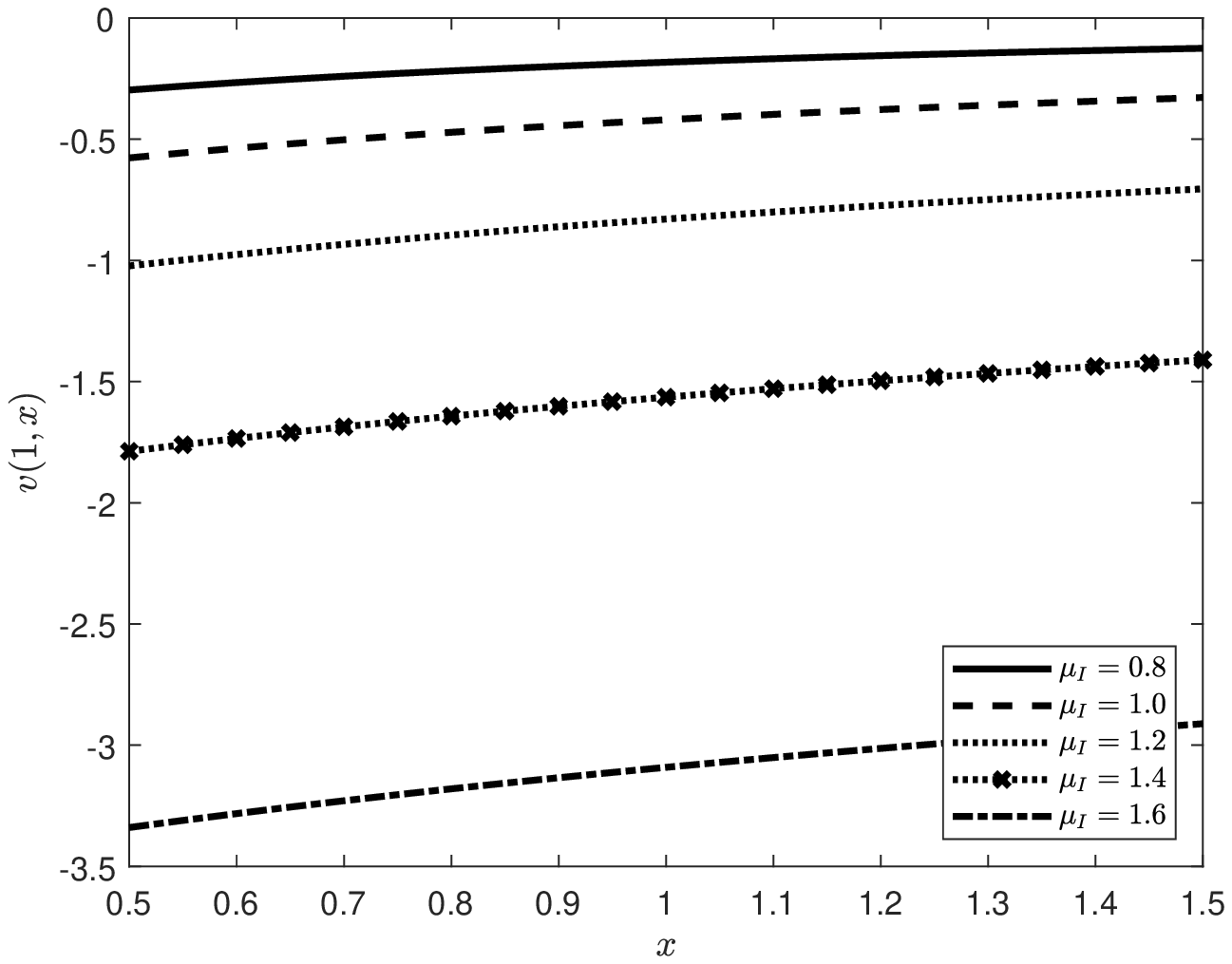}\includegraphics[height=2.1in]{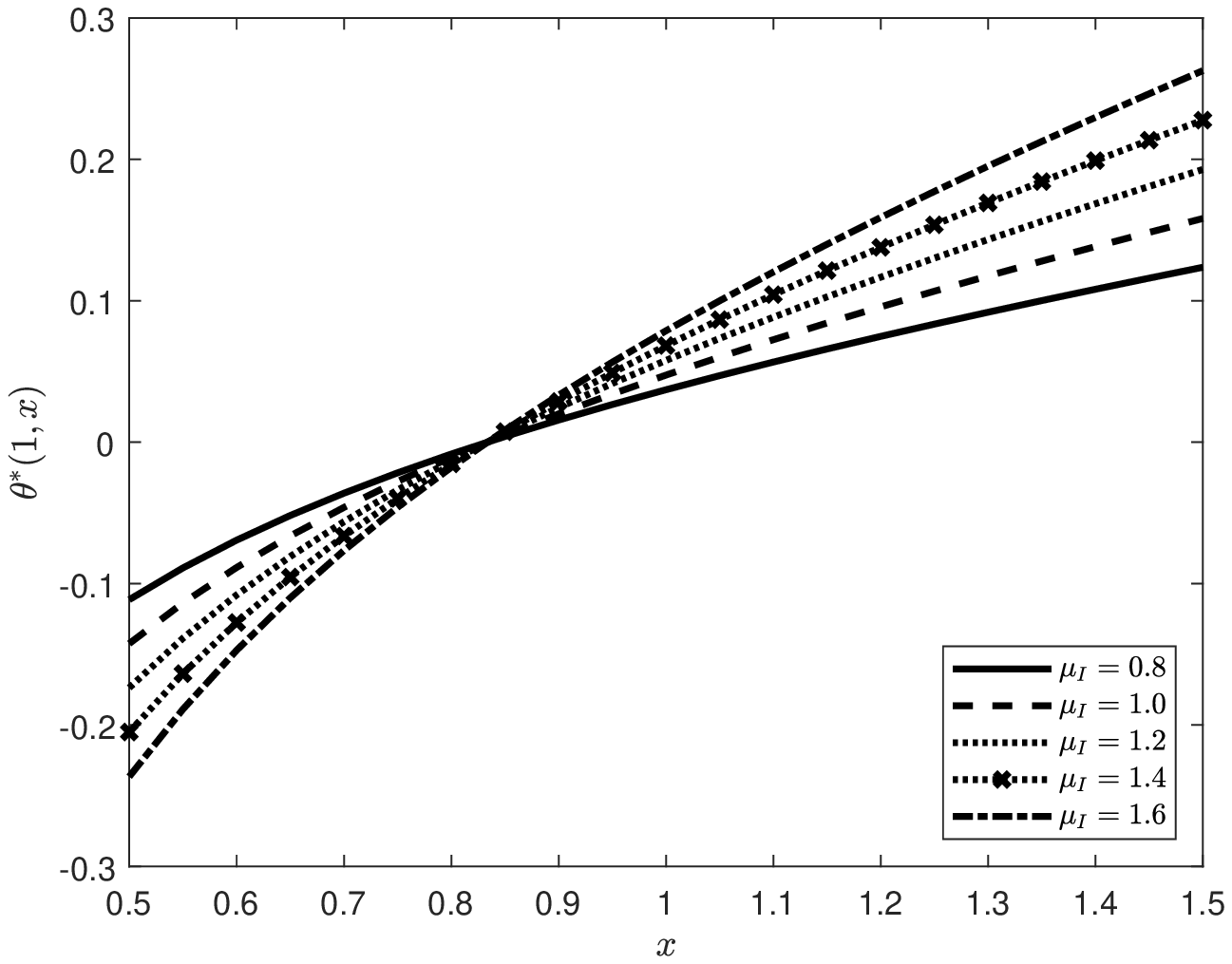}\\
\hspace{-0.2in}\vbox{Figure 1: {\small The value function $x\to v(1,x)$ (left panel); The optimal portfolio $x\to\theta^*(1,x)$ (right panel). In both panels, we fix $z=1$ and choose parameters $\mu=0.3$, $\sigma=1$, $\sigma_I=0.25$ and $\rho=2$.}}
\end{array}
\end{array}
$$

In Figure 2, we plot graphs with the changing index volatility $\sigma_I$. For all wealth level, the value function $v(1,x)$ is increasing in $\sigma_I$ and the optimal portfolio $\theta^*(1,x)$ is decreasing in $\sigma_I$. The right panel illustrates that the fund manager will become more conservative to invest in the risky asset if the index has high volatility because the difference between the index process and the portfolio account becomes highly volatile, which may cause very high and frequent capital injections. In this case, as $\sigma_I$ increases, the fund manager will prefer to invest less and use the cash in the tracking procedure to reduce the frequency of large fluctuations. It is interesting to see that the resulting total capital injection can be reduced by the increase in index volatility as shown in the left panel. 

$$
\begin{array}{ccc}
\begin{array}{c}
\includegraphics[height=2.1in]{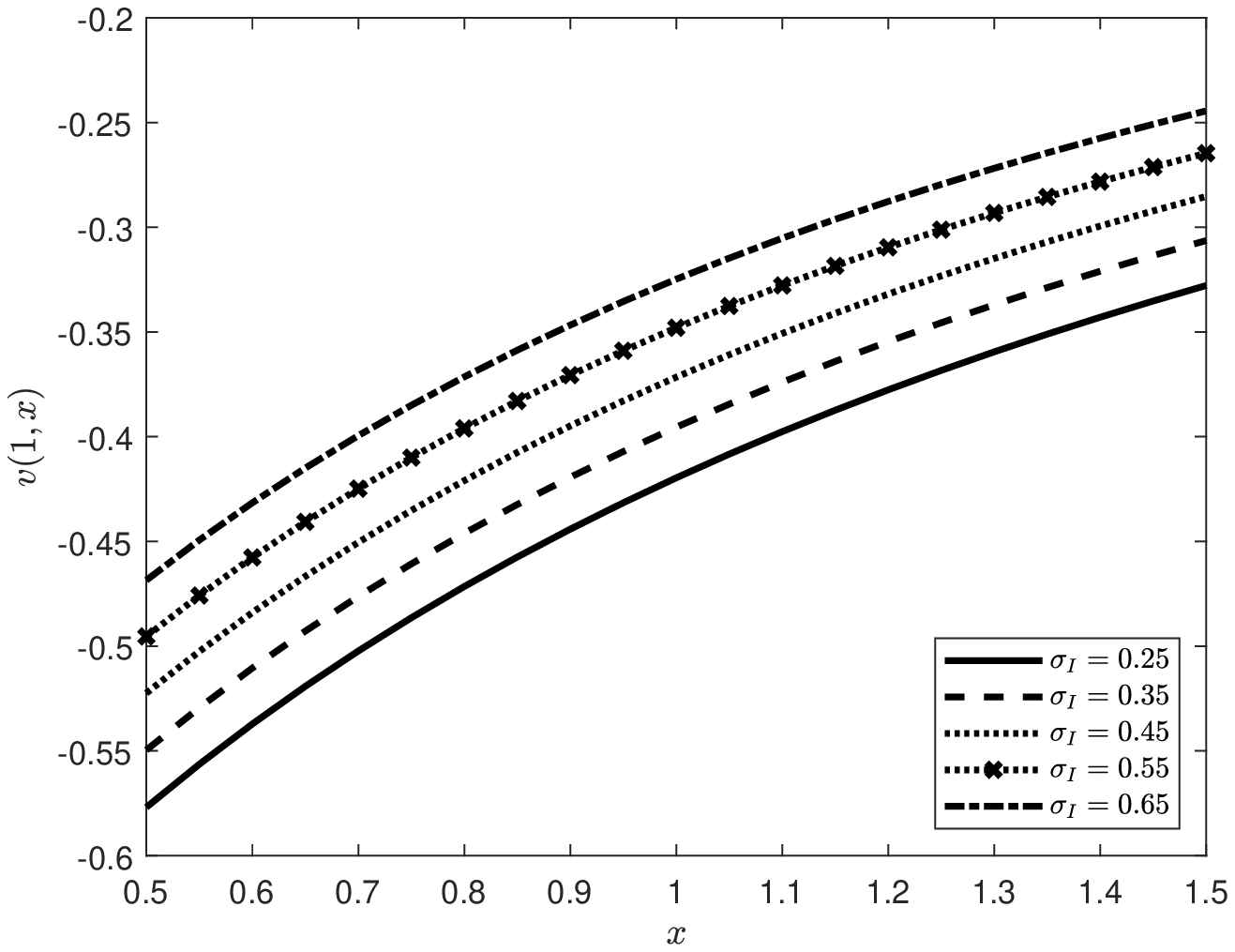}\includegraphics[height=2.1in]{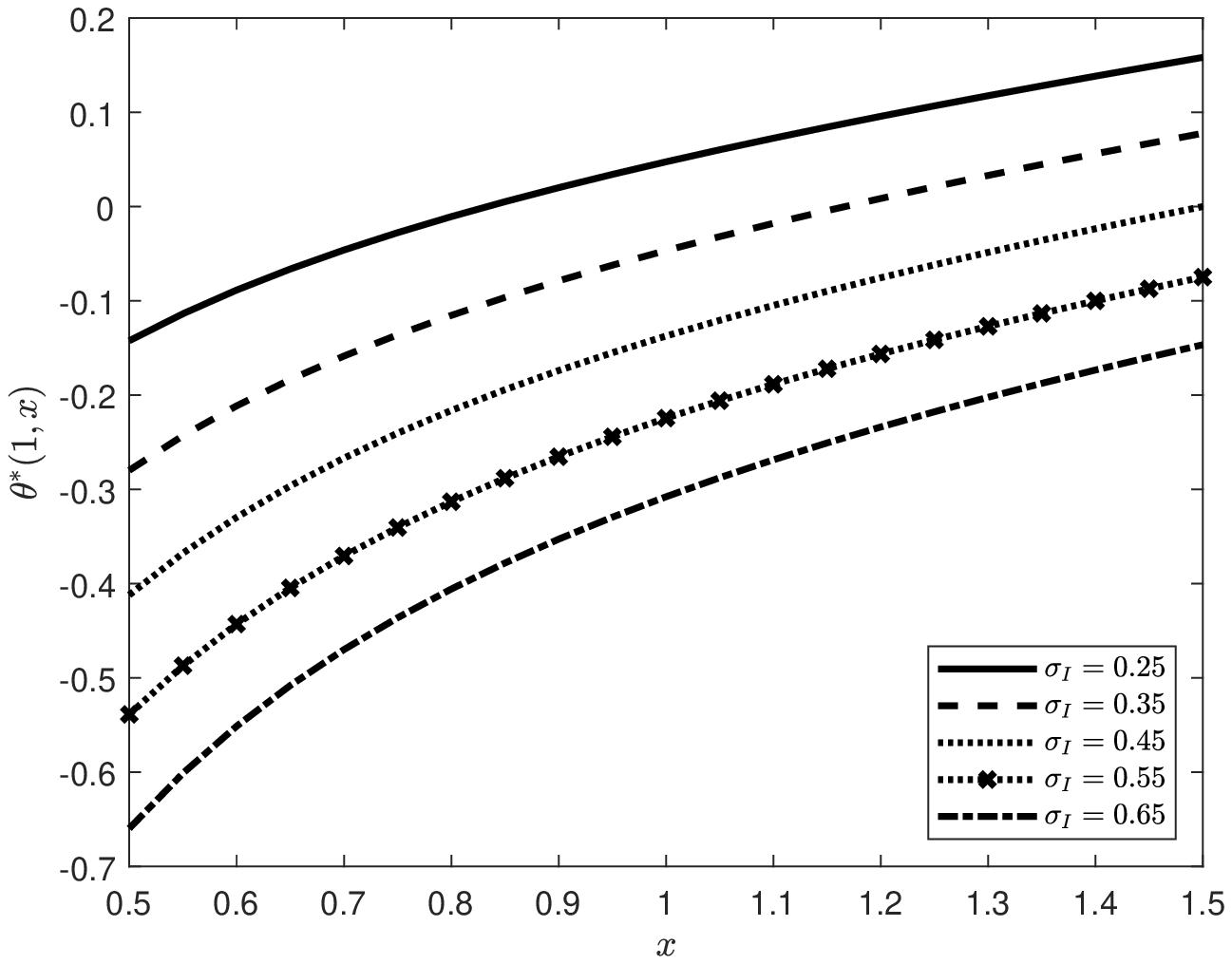}\\
\hspace{-0.2in}\vbox{Figure 2: {\small The value function $x\to v(1,x)$ (left panel); The optimal portfolio $x\to\theta^*(1,x)$ (right panel). In both panels, we fix $z=1$ and choose parameters $\mu=0.3$, $\sigma=1$, $\mu_I=1$ and $\rho=2$.}}
\end{array}
\end{array}
$$

We also plot in Figure 3 and Figure 4 the sensitivity results with respect to the risky asset return $\mu$ and volatility $\sigma$. One can see that both the value function and the optimal portfolio are increasing in $\mu$ and decreasing in $\sigma$, which may help to explain some real life observations. Figure 3 illustrates that the higher return of the risky asset will incentivitize the fund manager to invest more aggressively so that the portfolio performance can beat the index benchmark more often and the external capital injection can be reduced. In Figure 4, the high volatility in risky asset may cause frequent loss in the portfolio account, and the fund manager becomes more hesitant to invest in the risky asset to avoid the possibly frequent capital injections. The overall capital injection is lifted up by the larger volatility $\sigma$ as shown in the left panel. Comparing with the left panel of Figure 2, we can see that the volatility of the index and the volatility of the risky asset have opposite impacts on the total injected capital under our tracking procedure.

$$
\begin{array}{ccc}
\begin{array}{c}
\includegraphics[height=2.1in]{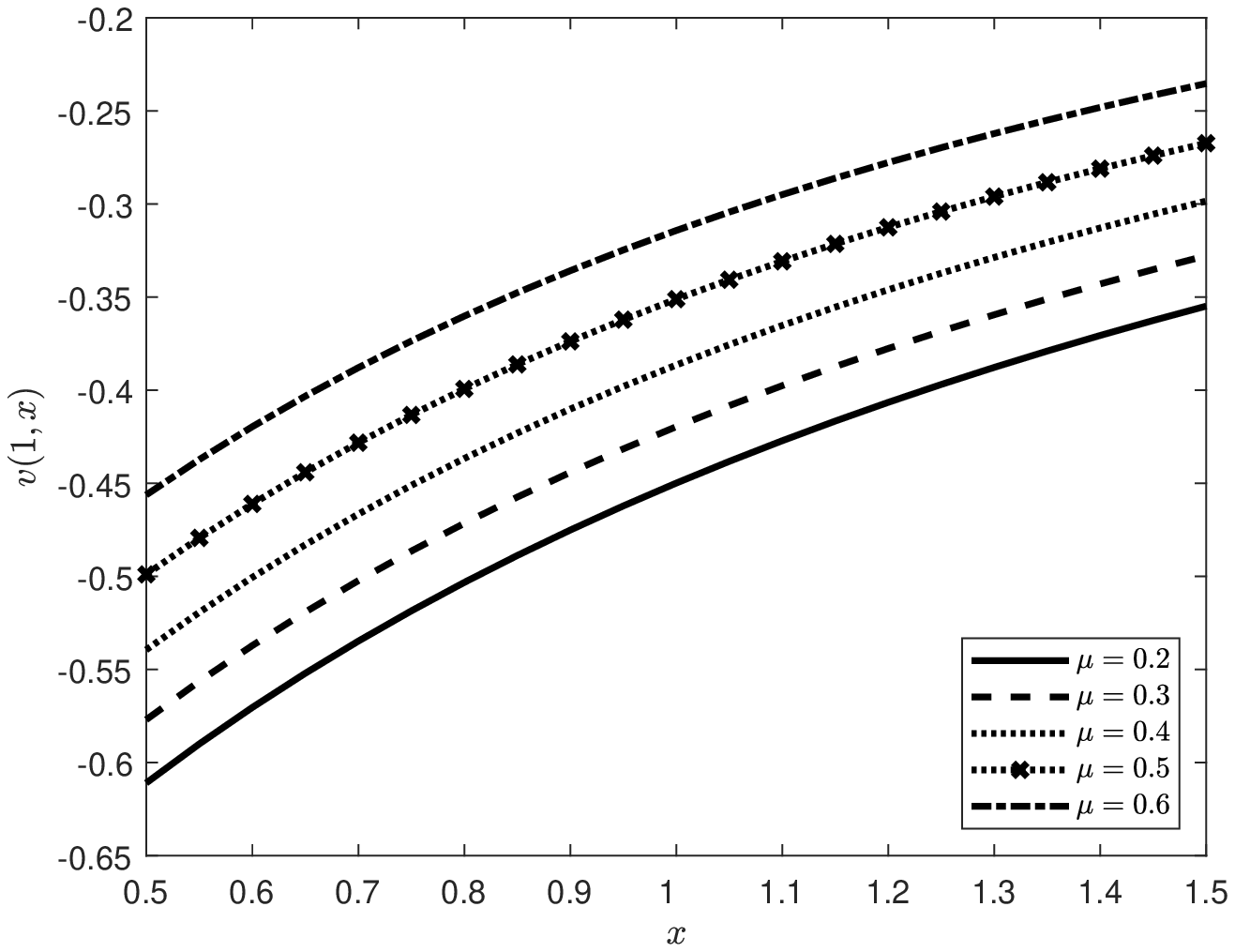}\includegraphics[height=2.1in]{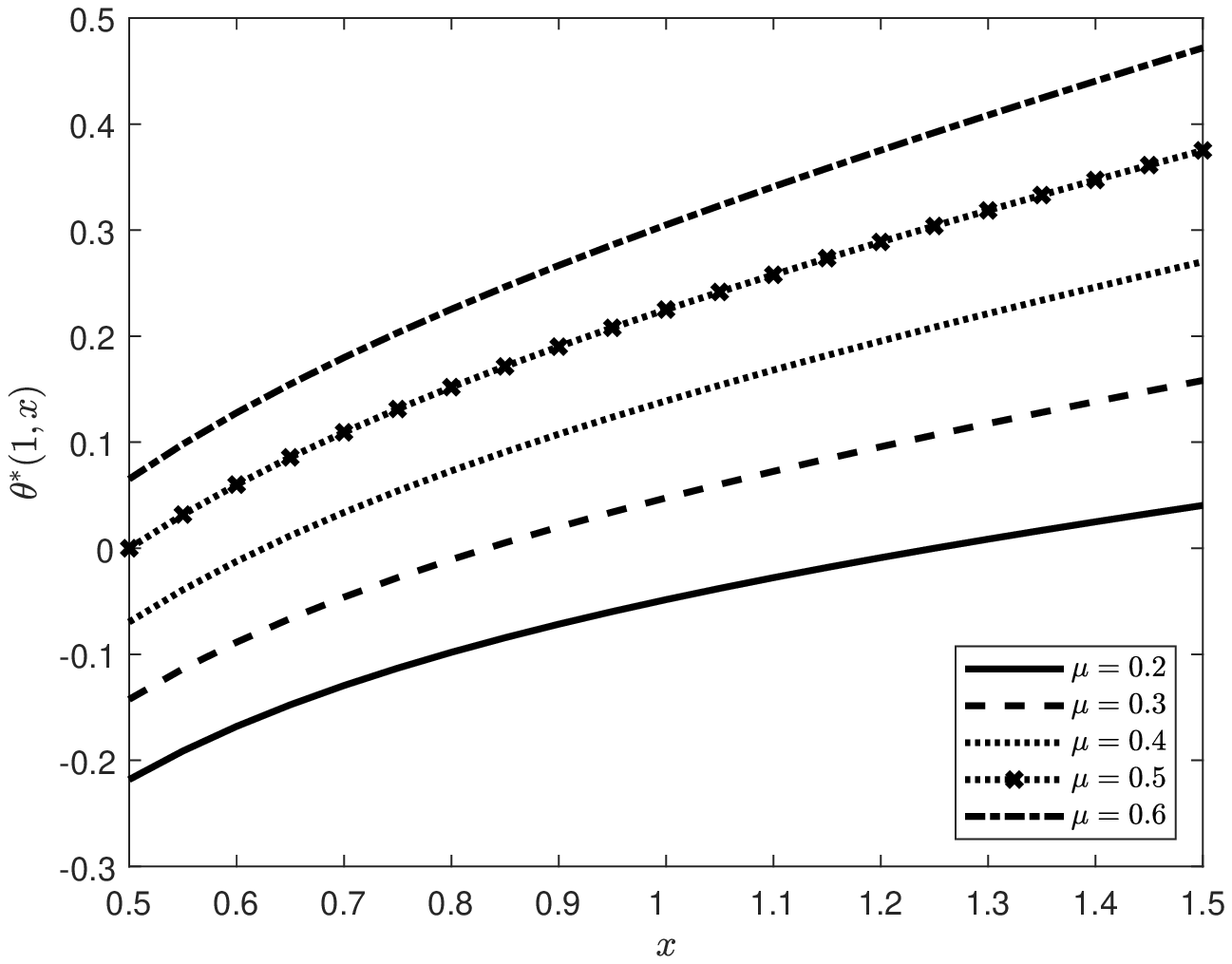}\\
\hspace{-0.2in}\vbox{Figure 3: {\small The value function $x\to v(1,x)$ (left panel); The optimal portfolio $x\to\theta^*(1,x)$ (right panel). In both panels, we fix $z=1$ and choose parameters $\mu_I=1$, $\sigma_I=0.25$, $\sigma=1$ and $\rho=2$.}}
\end{array}
\end{array}
$$

$$
\begin{array}{ccc}
\begin{array}{c}
\includegraphics[height=2.1in]{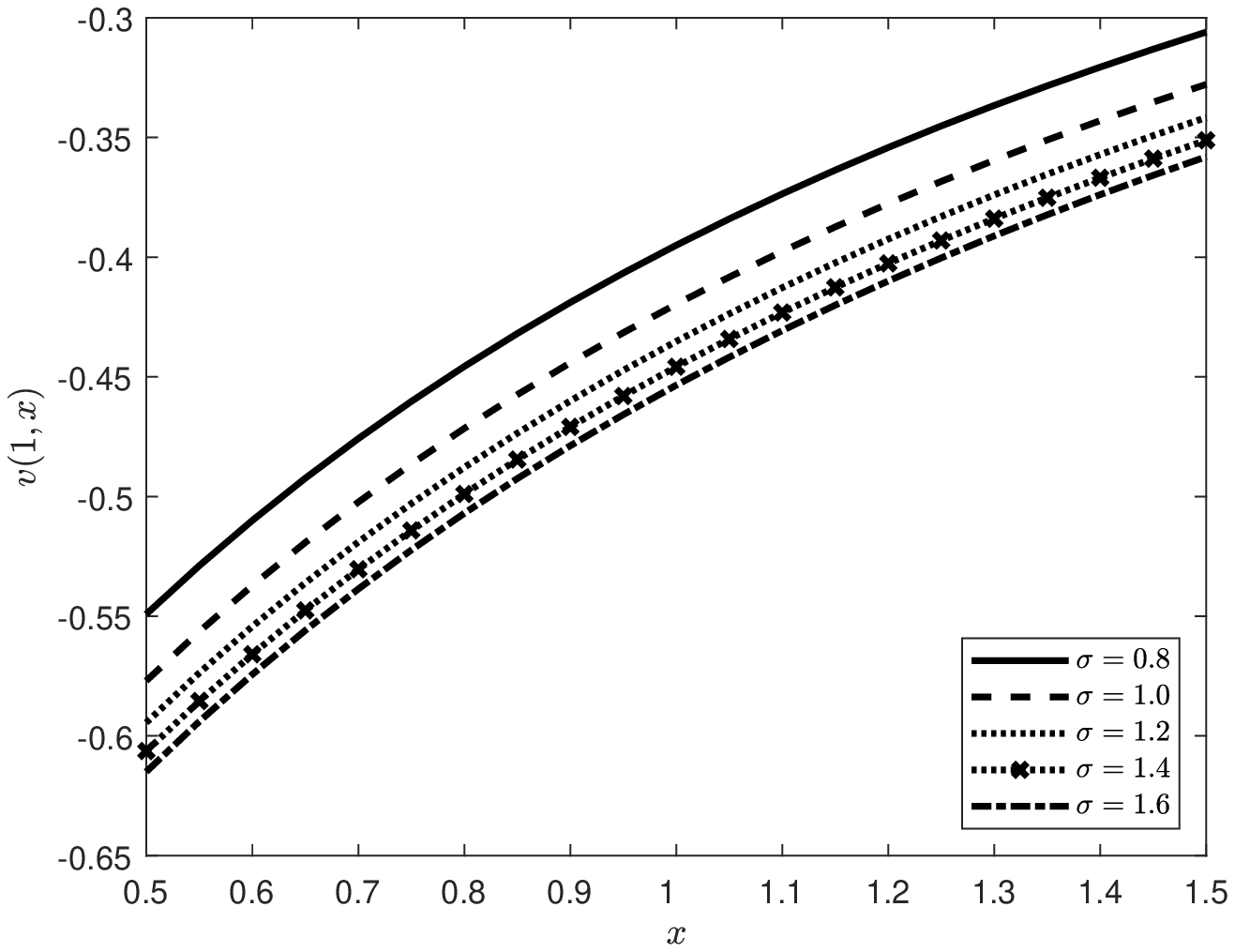}\includegraphics[height=2.1in]{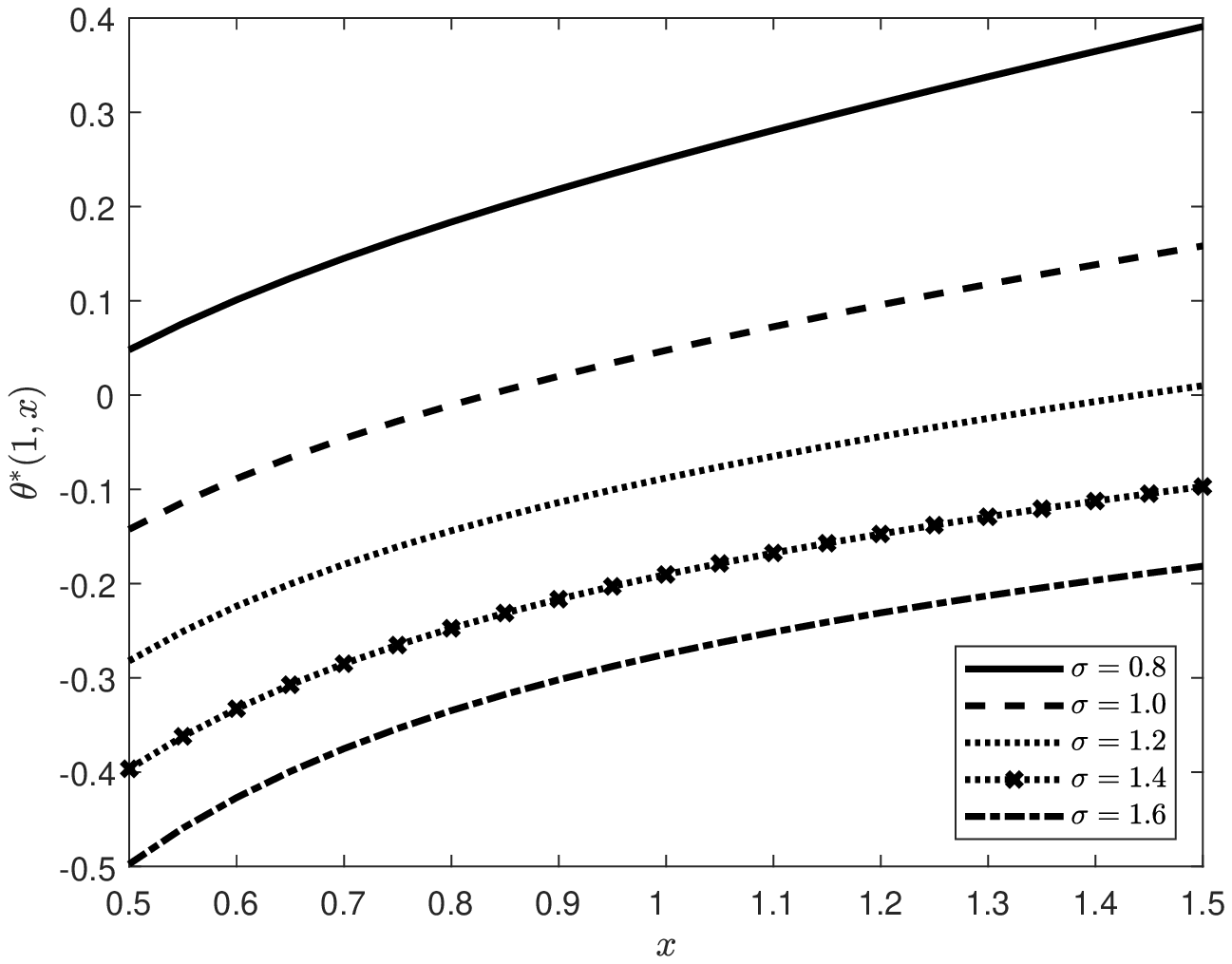}\\
\hspace{-0.2in}\vbox{Figure 4: {\small The value function $x\to v(1,x)$ (left panel); The optimal portfolio $x\to\theta^*(1,x)$ (right panel). In both panels, we fix $z=1$ and choose parameters $\mu_I=1$, $\sigma_I=0.25$, $\mu=0.3$ and $\rho=2$.}}
\end{array}
\end{array}
$$

At last, similar to the finite horizon case, we can also consider the simple case when $\sigma_I=0$, in which the deterministic benchmark process $I_t=ze^{\mu_It}$ describes a growth rate as studied in  \cite{YaoZZ06}. By \eqref{eq:clode-formsolHJBGBM} and \eqref{eq:barthetastar}, it follows that the value function and the optimal feedback strategy can be further simplified to
\begin{align}\label{eq:casesigmasp0}
v(z,x)&=z\frac{\gamma_0-1}{\gamma_0}\left(1+\frac{x}{z}\right)^{\frac{\gamma_0}{\gamma_0-1}},\quad \bar{\theta}^*(z,x)=-(\gamma_0-1)(x+z)(\sigma\sigma^{\top})^{-1}\mu,
\end{align}
where the constant $\gamma_0>0$ is given by
\begin{align*}
\gamma_0=\frac{\alpha-\rho+\sqrt{(\rho-\alpha)^2+4\alpha(\rho-\mu_I)}}{2\alpha}.
\end{align*}
and $\alpha$ is defined in \eqref{eq:alpha}.

\section{Proofs of Main Results}\label{app:proof1}
This section collects the proofs of some main results in the previous sections.

\begin{proof}[Proof of Lemma \ref{first-transf}]
It is clear that, for $(a,{\rm v},z)\in[0,\infty)^2\times\R$,
\begin{align*}
u(a,\mathrm{v},z)=\inf_{\theta} \inf_{C} \mathbb{E}\left[C_0+\int_0^Te^{-\rho t}dC_t \right],
\end{align*}
and the integration by parts gives that $C_0+\int_0^Te^{-\rho t}dC_t =e^{-\rho T}C_T +\rho \int_0^T e^{-\rho t}C_tdt$.
For each fixed $\theta=(\theta_t)_{t\in[0,T]}$, we need to choose the optimal singular control $C=(C_t)_{t\in[0,T]}$ to minimize
\begin{align*}
\inf_{C} F(C),\ \ \text{where}\ \ F(C):=\mathbb{E}\left[e^{-\rho T}C_T +\rho \int_0^T e^{-\rho t}C_tdt \right],
\end{align*}
subjecting to $C_t\geq A_t-V_t^{\theta}$ at each $t\in[0,T]$. Note that the cost functional $F(C)$ is strictly increasing in $C$. That is, if $C^1 \le  C^2$ and $C^1\ne  C^2$, then we have $F(C^1)<F(C^2)$. Therefore, the optimal choice of the control $C$ is the minimal non-negative and non-decreasing process $C_t$ such that $C_t\geq A_t-V_t^{\theta}$ for $t\in[0,T]$. We claim that the minimal process is the non-decreasing envelope $C^*_t:=0\vee \sup_{s\leq t} (A_s-V_s^{\theta})$. Note that $C_t^*$ is non-negative and satisfies the dynamic floor constraint. Let $\widetilde{C}$ be another non-negative and non-decreasing process satisfying $\widetilde{C}_t\geq A_t-V_t^{\theta}$, $t\in[0,T]$. Suppose that $\widetilde{C}\leq C^*$ and $\widetilde{C}\ne C^*$. That is, there exists a $t\in[0,T]$ and a set $O$ with $\mathbb{P}(O)>0$ such that $\widetilde{C}_t(\omega)<C_t(\omega)$ for $\omega\in O$. By definition, we have
$$A_t(\omega)-V_t^{\theta}(\omega)\leq \widetilde{C}_t(\omega)<C_t(\omega)=\sup_{s\leq t}(A_s(\omega)-V_s^{\theta}(\omega)).$$
For each fixed $\omega\in O$, let $t^*<t$ be the time such that $A_{t^*}(\omega)-V_{t^*}^{\theta}(\omega)=\sup_{s\leq t}(A_s(\omega)-V_s^{\theta}(\omega))$. It follows that $\widetilde{C}_t(\omega)<A_{t^*}(\omega)-V_{t^*}^{\theta}(\omega)\leq \widetilde{C}_{t^*}(\omega)$.
We obtain a contradiction that the process $\widetilde{C}$ is non-decreasing. Therefore, the original problem can be written as
\begin{align*}
u(a,\mathrm{v},z)&=C^*_0+\inf_{\theta}\mathbb{E}\left[\int_0^Te^{-\rho t}dC_t^* \right]=(a-{\rm v})^++ \inf_{\theta}\mathbb{E}\left[\int_0^Te^{-\rho t}d(0\vee \sup_{s\leq t}(A_s-V_s^{\theta}) \right],
\end{align*}
which completes the proof.
\end{proof}

\begin{proof}[Proof of Proposition~\ref{prop:smoothh}] We first derive the representation of the partial derivative $h_u$ of the function $h$ w.r.t. the variable $u$. Let $(t,z)\in[0,T]\times\R$ be fixed. For any $u_2>u_1\geq0$, it follows by \eqref{eq:dualtransfromh} that
\begin{align*}
  \frac{h(t,z,u_2)-h(t,z,u_1)}{u_2-u_1}=-\int_t^T\Ex\left[e^{-\rho s}f(s,M_s^{t,z})\frac{e^{-R_s^{t,u_2}}-e^{-R_s^{t,u_1}}}{u_2-u_1}\right]ds.
\end{align*}
A direct calculation yields that, for $s\in[t,T]$,
\begin{align*}
  \lim_{u_2\downarrow u_1}\frac{e^{-R_s^{t,u_2}-\rho s}-e^{-R_s^{t,u_1}-\rho s}}{u_2-u_1}=\begin{cases}
    -e^{-R_s^{t,u_1}-\rho s},\quad\max_{r\in[t,s]}\left[-\sqrt{2\alpha}(B^1_r-B^1_t)-(\alpha-\rho)(r-t)\right]\leq u_1, \\ \\
    \qquad~0,\qquad\quad~\max_{r\in[t,s]}\left[-\sqrt{2\alpha}(B^1_r-B^1_t)-(\alpha-\rho)(r-t)\right]>u_1.
  \end{cases}
\end{align*}
As $\sup_{(s,u_1,u_2)\in[t,T]\times[0,\infty)^2}\left|\frac{e^{-R_s^{t,u_2}-\rho s}-e^{-R_s^{t,u_1}-\rho s}}{u_2-u_1}\right|\leq1$, the dominated convergence theorem gives that
\begin{align}\label{eq:derihu1}
\lim_{u_2\downarrow u_1}\frac{h(t,z,u_2)-h(t,z,u_1)}{u_2-u_1}&=\Ex\left[\int_t^Te^{-\rho s}f(s,M_s^{t,z})e^{-R_s^{t,u_1}}\mathbf{1}_{\left\{\max_{r\in[t,s]}\left[-\sqrt{2\alpha}(B^1_r-B^1_t)-(\alpha-\rho)(r-t)\right]\leq u_1\right\}}ds\right]\notag\\
  &=\Ex\left[\int_t^{\tau_{u_1}^t\wedge T} e^{-\rho s}f(s,M_s^{t,z})e^{-R_s^{t,u_1}}ds\right],
\end{align}
where $\tau_{u_1}^t:=\inf\{s\geq t;~-\sqrt{2\alpha}(B^1_s-B^1_t)-(\alpha-\rho)(s-t)=u_1\}$. On the other hand, for the case $u_1>u_2\geq0$, similar to the computation of \eqref{eq:derihu1}, we have that $\lim_{u_2\uparrow u_1}\frac{h(t,z,u_2)-h(t,z,u_1)}{u_2-u_1}=\lim_{u_2\downarrow u_1}\frac{h(t,z,u_2)-h(t,z,u_1)}{u_2-u_1}$. Therefore, the representation~\eqref{eq:derihu} holds.

We next derive the representation of $h_{uu}$. Let $(t,z)\in[0,T]\times\R$ be fixed. For any $u_0, u_n\geq 0$ and $u_n\to u_0$ as $n\to\infty$, we have that, for $n\geq1$,
\begin{align}\label{eq:Deltandecom}
\Delta_n&=\frac{h_u(t,z,u_n)-h_u(t,z,u_0)}{u_n-u_0}=\Ex\left[\frac{1}{u_n-u_{0}}\int_{\tau_{0}}^{\tau_{n}} e^{-\rho s}f(s,M_s^{t,z})e^{-R_s^{t,u_0}}ds\right]\nonumber\\
&\quad+\Ex\left[\frac1{u_n-u_0}\int_t^{\tau_0}e^{-\rho s}f(s,M_s^{t,z})\left(e^{-R_s^{t,u_n}}-e^{-R_s^{t,u_0}}\right)ds\right]\notag\\
  &\quad+\Ex\left[\frac1{u_n-u_0}\int_{\tau_0}^{\tau_n}e^{-\rho s}f(s,M_s^{t,z})\left(e^{-R_s^{t,u_n}}-e^{-R_s^{t,u_0}}\right)ds\right]
  :=\Delta^{(1)}_n+\Delta^{(2)}_n+\Delta^{(3)}_n.
\end{align}
where $\tau_0:=\tau_{u_0}^t\wedge T$ and $\tau_n:=\tau_{u_n}^t\wedge T$. In order to deal with $\Delta_n^{(1)}$, we first focus on the case when $u_n\downarrow u_0$ as $n\to\infty$. Let us introduce $\tilde{\Delta}^{(1)}_n:=\Ex\left[\frac{\tau_n-\tau_0}{u_n-u_0} e^{-\rho \tau_0}f(\tau_0,M_{\tau_0}^{t,y})e^{-R_{\tau_0}^{t,u_0}}\right]$. For any $m>0$, it follows from the assumption {\Af} that
\begin{align}\label{eq:Deltadifference}
|\Delta^{(1)}_n-\tilde{\Delta}^{(1)}_n|&\!\!\leq\!\!\Ex\left[\frac{\tau_n-\tau_0}{u_n-u_0}\xi_n\right]
\!\!\leq\!\! m\Ex\left[\xi_n\right]\!\!+\!C\left\{1+\Ex\left[\max_{s\in[t,T]}\left|M^{t,z}_s\right|^2\right]^{\frac12}\right\}\Px\left(\frac{\tau_n-\tau_0}{u_n-u_0}>m\right)^{\frac12},
\end{align}
where for $n\geq1$,
\begin{align*}
  \xi_n:=\max_{s\in[\tau_0,\tau_n]}\left|e^{-\rho s}f(s,M_s^{t,z})e^{-R_s^{t,u}}-e^{-\rho \tau_0}f(\tau_0,M_{\tau_0}^{t,y})e^{-R_{\tau_0}^{t,u}}\right|.
\end{align*}
Note that $\xi_n\leq C(1+\sup_{s\in[t,T]}|M^{t,z}_s|)$ for all $n\geq1$ by the assumption {\Af}. For any $m>0$, it follows that $\tau_n\downarrow\tau_0$, $\Px$-a.s. as $n\rightarrow+\infty$. Therefore, we have that $\xi_n\downarrow0$, as $n\to\infty$, $\Px$-a.s., and hence $\Ex[\xi_n]\rightarrow0$ as $n\to\infty$. On the other hand, by setting $\tilde{\mu}:=\alpha-\rho$, we have that
\begin{align*}
\Px\left(\frac{\tau_n-\tau_0}{u_n-u_0}>m\right)\leq\int_{m(u_n-u_0)}^{+\infty}\frac{u_n-u_0}{\sqrt{4\alpha\pi t^3}}e^{-\frac{(u_n-u_0-\tilde{\mu}t)^2}{4\alpha t}}dt\leq\sqrt{u_n-u_0}\int_{m}^{+\infty}\frac{1}{\sqrt{4\alpha\pi s^3}}ds.
\end{align*}
Letting $n\rightarrow+\infty$ in \eqref{eq:Deltadifference}, we arrive at
\begin{align}\label{eq:limDeltadifference}
  \lim_{n\rightarrow+\infty}|\Delta^{(1)}_n-\tilde{\Delta}^{(1)}_n|=0.
\end{align}
Moreover, using the strong Markov property of Brownian motion with drift, it follows that
\begin{align*}
  \Ex\left[\frac{\tau_n-\tau_0}{u_n-u_0}\bigg|\mathcal{F}_{\tau_0+}\right]
  &=\int_0^{T-\tau_0}\frac{1}{\sqrt{4\alpha\pi s}}e^{-\frac{(u_n-u_0-\tilde{\mu}s)^2}{4\alpha s}}ds+(T-\tau_0)\int_{T-\tau_0}^{+\infty}\frac1{\sqrt{4\alpha\pi s^3}}e^{-\frac{(u_n-u_0-\tilde{\mu}s)^2}{4\alpha s}}ds.
\end{align*}
Therefore, for $n\geq1$,
\begin{align*}
\tilde{\Delta}^{(1)}_n&=\Ex\left[e^{-\rho \tau_0}f(\tau_0,M_{\tau_0}^{t,z})e^{-R_{\tau_0}^{t,u_0}}\int_0^{T-\tau_0}\frac{1}{\sqrt{4\alpha\pi s}}e^{-\frac{(u_n-u_0-\tilde{\mu}s)^2}{4\alpha s}}ds\right]\notag\\
  &\quad+\Ex\left[(T-\tau_0)e^{-\rho \tau_0}f(\tau_0,M_{\tau_0}^{t,y})e^{-R_{\tau_0}^{t,u_0}}\int_{T-\tau_0}^{+\infty}\frac1{\sqrt{4\alpha\pi s^3}}e^{-\frac{(u_n-u_0-\tilde{\mu}s)^2}{4\alpha s}}ds\right].
\end{align*}
This yields that
\begin{align*}
\lim_{n\rightarrow+\infty}\Delta^{(1)}_n=\lim_{n\rightarrow+\infty}\tilde{\Delta}^{(1)}_n=\Ex\left[e^{-\rho \tau_0}f(\tau_0,M_{\tau_0}^{t,z})e^{-R_{\tau_0}^{t,u_0}}\Gamma(\tau_0)\right],
\end{align*}
where, for $t\in[0,T]$,
\begin{align}\label{eq:Gammat}
  \Gamma(t):=\int_0^{T-t}\frac{1}{\sqrt{4\alpha\pi s}}e^{-\frac{\tilde{\mu}^2}{4\alpha}s}ds+(T-t)\int_{T-t}^{+\infty}\frac1{\sqrt{4\alpha\pi s^3}}e^{-\frac{\tilde{\mu}^2}{4\alpha}s}ds.
\end{align}

For the case where $u_0>0$, $u_n\uparrow u_0$ as $n\to\infty$, we can follow the similar argument to get that
\begin{align*}
  \lim_{n\rightarrow+\infty}\Delta^{(1)}_n&=\lim_{n\rightarrow+\infty}\Ex\left[e^{-\rho \tau_n}f(\tau_n,M_{\tau_n}^{t,z})e^{-R_{\tau_n}^{t,u_n}}\int_0^{T-\tau_n}\frac{1}{\sqrt{4\alpha\pi s}}e^{-\frac{(u_n-u_0-\tilde{\mu}s)^2}{4\alpha s}}ds\right]\notag\\
  &\quad+\lim_{n\rightarrow+\infty}\Ex\left[(T-\tau_n)e^{-\rho \tau_n}f(\tau_n,M_{\tau_n}^{t,z})e^{-R_{\tau_n}^{t,u_n}}\int_{T-\tau_n}^{+\infty}\frac1{\sqrt{4\alpha\pi s^3}}e^{-\frac{(u_n-u_0-\tilde{\mu}s)^2}{4\alpha s}}ds\right]\notag\\
  &=\Ex\left[e^{-\rho \tau_0}f(\tau_0,M_{\tau_0}^{t,z})e^{-R_{\tau_0}^{t,u_0}}\Gamma(\tau_0)\right].
\end{align*}
Similar to the derivation of \eqref{eq:derihu}, we also have that
\begin{align*}
  \lim_{n\rightarrow+\infty}\Delta^{(2)}_n&=\lim_{n\rightarrow+\infty}\Ex\left[\frac1{u_n-u_0}\int_t^{\tau_0}e^{-\rho s}f(s,M_s^{t,z})\left(e^{-R_s^{t,u_n}}-e^{-R_s^{t,u_0}}\right)ds\right]\notag\\
  &=\Ex\left[\int_t^{\tau_0}e^{-\rho s}f(s,M_s^{t,z})e^{-R_s^{t,u_0}}ds\right].
\end{align*}
Finally, similar to \eqref{eq:Deltadifference} and \eqref{eq:limDeltadifference}, we can show by the assumption {\Af} that
\begin{align*}
|\Delta^{(3)}_n|&\leq\Ex\left[\frac1{u_n-u_0}\int_{\tau_0}^{\tau_n}e^{-\rho s}f(s,M_s^{t,z})\left|e^{-R_s^{t,u_n}}-e^{-R_s^{t,u_0}}\right|ds\right]\notag\\
&\leq m\Ex[\tilde{\xi}_n]+C\left\{1+\mathbb{E}\left[\max_{s\in[t,T]}\left|M^{t,z}_s\right|^{2}\right]\right\}\Px\left(\frac{\tau_n-\tau_0}{u_n-u_0}>m\right)^{\frac12},
\end{align*}
where, for $t\in[0,T]$,
\begin{align*}
  \tilde{\xi}_n&:=\max_{s\in[\tau_0,\tau_n]}\left|e^{-\rho s}f(s,M_s^{t,z})e^{-R_s^{t,u_n}}-e^{-\rho s}f(s,M_s^{t,z})e^{-R_s^{t,u_0}}\right|.
\end{align*}
By the assumption {\Af}, we have $\Ex[\tilde{\xi}_n]\to0$ as $n\to\infty$. Hence $\lim_{n\rightarrow+\infty}|\Delta^{(3)}_n|=0$. Putting all the pieces together, we can derive from the decomposition~\eqref{eq:Deltandecom} and $R_{\tau_0}^{t,u_0}=0$ that
\begin{align}\label{eq:Rephuu}
h_{uu}(t,z,u_0) 
&=\Ex\left[e^{-\rho \tau_0}f(\tau_0,M_{\tau_0}^{t,z})\Gamma(\tau_0)\right]
+\Ex\left[\int_t^{\tau_0}e^{-\rho s}f(s,M_s^{t,z})e^{-R_s^{t,u_0}}ds\right].
\end{align}
where we define that $\tau_0:=\tau_{u_0}^t\wedge T$, and $\Gamma(t)$ for $t\in[0,T]$ is given by \eqref{eq:Gammat}.

We next derive the representations of $h_{y}$, $h_{yy}$ and $h_{yu}$. Fix $(t,u)\in[0,T]\times[0,\infty)$. In view of the assumption {\AZ}, Theorem 3.3.2 in \cite{Kunita19} yields that, for $s\in[t,T]$, the family $(M^{t,z}_s)_{z\in\R}$ admits a modification which is continuously differentiable w.r.t. $z$. Moreover, $\partial_z M^{t,z}_s$ is continuous in $z$ and satisfies the following SDE for $s\in[t,T]$ that
\begin{align}\label{eq:dynamic-partial-Y}
\partial_z M^{t,z}_s&=1+\int_t^s\mu'_Z\big(M^{t,z}_r\big)\partial_zM^{t,z}_rdr+\varrho\int_t^s\sigma'_Z(M_r^{t,z})\partial_z M^{t,z}_rdB^1_r\nonumber\\
  &\quad+\sqrt{1-\varrho^2}\int_t^s\sigma'_Z(M_r^{t,z})\partial_z M^{t,z}_rdB^2_r,
\end{align}
and for any $p\geq2$, the following moment estimate holds that
\begin{align}\label{eq:momentDeriYty}
  \sup_{z\in\R}\Ex\left[\max_{s\in[t,T]}\left|\partial_z M^{t,z}_s\right|^p\right]<+\infty.
\end{align}
Thanks to \eqref{eq:dualtransfromh}, for distinct $z,\hat{z}\in\R$ and some constant $C>0$, we have that
\begin{align}\label{eq:diff-der-y}
  \frac{h(t,z,u)-h(t,\hat{z},u)}{z-\hat{z}}=-\Ex\left[\int_t^T e^{-\rho s-R_s^{t,u}}\frac{f(s,M_s^{t,z})-f(s,M_s^{t,\hat{z}})}{z-\hat{z}}ds\right].
\end{align}
By the assumption {\Af}, for $s\in[t,T]$, the next results hold $\Px$-a.s. that
\begin{align*}
\begin{cases}
\displaystyle \frac{f(s,M_s^{t,z})-f(s,M_s^{t,\hat{z}})}{z-\hat{z}}\overset{\hat{z}\to z}{\longrightarrow} f'(s,M_s^{t,z})\partial_zM_s^{t,z},\\ \\
\displaystyle \left|\frac{f(s,M_s^{t,z})-f(s,M_s^{t,\hat{z}})}{z-\hat{z}}\right|\leq C\left|\frac{M_s^{t,z}-M_s^{t,\hat{z}}}{z-\hat{z}}\right|.
\end{cases}
\end{align*}
We have from \eqref{eq:momentDeriYty} that, for any $p\geq2$, $\sup_{\hat{z}\neq z}\Ex\left[\left|\frac{M_s^{t,z}-M_s^{t,\hat{z}}}{z-\hat{z}}\right|^p\right]<+\infty$. This implies that $(\frac{M_s^{t,z}-M_s^{t,\hat{z}}}{z-\hat{z}})_{\hat{z}\neq z}$ is uniformly integrable. Therefore, in view of \eqref{eq:diff-der-y}, we arrive at
\begin{align}\label{eq:hy}
  h_z(t,z,u)&=-\lim_{\hat{z}\rightarrow z}\Ex\left[\int_t^T e^{-\rho s-R_s^{t,u}}\frac{f(s,M_s^{t,z})-f(s,M_s^{t,\hat{z}})}{z-\hat{z}}ds\right]\nonumber\\
  &=-\Ex\left[\int_t^T e^{-\rho s-R_s^{t,u}}f'(s,M_s^{t,z})\partial_zM_s^{t,z}ds\right],
\end{align}
where $f'(t,z)$ denotes the partial derivative of $f$ w.r.t. $z$. Similarly, we can obtain that
\begin{align}\label{eq:hyu}
h_{zu}(t,z,u)&=\Ex\left[\int_t^Te^{-\rho s-R_s^{t,u}}f'(s,M_s^{t,z})\partial_z M_s^{t,z}\mathbf{1}_{\left\{\max_{r\in[t,s]}\left[-\sqrt{2\alpha}(B^1_r-B^1_t)-(\alpha-\rho)(r-t)\right]\leq u\right\}}ds\right].
\end{align}

We next derive the expression of $h_{yy}$. To this end, we need the dynamics of $\partial^2_{zz} M^{t,z}_s$ for $s\in[t,T]$. Following a similar argument of Theorem 3.4.2 of \cite{Kunita19}, we can deduce that
\begin{align*}
  \partial^2_{zz} M^{t,z}_s&=\int_t^s\left\{\mu''_Z(M^{t,z}_r)|\partial_zM^{t,z}_r|^2+\mu'_Z(M^{t,z}_r)\partial^2_{zz}M^{t,z}_r\right\}dr\notag\\
  &\quad+\varrho\int_t^s\left\{\sigma''_Z(M_r^{t,z})|\partial_z M^{t,z}_r|^2+\sigma'_Z(M_r^{t,z})\partial^2_{zz} M^{t,z}_r\right\}dB^1_r\notag\\
  &\quad+\sqrt{1-\varrho^2}\int_t^s\left\{\sigma''_Z(M_r^{t,z})|\partial_z M^{t,z}_r|^2+\sigma'_Z(M_r^{t,z})\partial^2_{zz} M^{t,z}_r\right\}dB^2_r,
\end{align*}
and it holds for $p\geq1$ that
\begin{align*}
  \sup_{z\in\R}\Ex\left[\max_{s\in[t,T]}\left|\partial^2_{zz} M^{t,z}_s\right|^{2p}\right]<+\infty.
\end{align*}
The chain rule with the assumption {\Af} yields that, $\Px$-a.s.
\begin{align*}
  \frac{f'(s,M_s^{t,z})\partial_zM_s^{t,z}-f'(s,M_s^{t,\hat{z}})\partial_zM_s^{t,\hat{z}}}{z-\hat{z}}\overset{\hat{z}\to z}{\longrightarrow} f''(s,M_s^{t,z})\big|\partial_zM_s^{t,z}\big|^2+f'(s,M_s^{t,z})\partial^2_{zz}M_s^{t,z},
\end{align*}
and there exists a constant $C>0$ such that, for all $\hat{z}\neq z$,
\begin{align}\label{eq:uniformly-integrable}
  &\left|\frac{f'(s,M_s^{t,z})\partial_zM_s^{t,z}-f'(s,M_s^{t,\hat{z}})\partial_zM_s^{t,\hat{z}}}{z-\hat{z}}\right|\!\!\leq\!\! C\left|\frac{M_s^{t,z}-M_s^{t,\hat{z}}}{z-\hat{z}}\partial_zM_s^{t,z}\right|\!\!+\!C\left|\frac{\partial_zM_s^{t,z}-\partial_zM_s^{t,\hat{z}}}{z-\hat{z}}\right|.
\end{align}
Define $I(s;z,\hat{z}):=\Ex\left[\max_{r\in[t,s]}|\partial_zM_r^{t,z}-\partial_zM_r^{t,\hat{z}}|^{2p}\right]$ for $(s,z,\hat{z})\in[t,T]\times\R^2$. It follows from \eqref{eq:dynamic-partial-Y} and the assumption {\AZ} that, for some $C>0$,
\begin{align*}
  I(s;z,\hat{z})&\leq C\int_t^s\left\{I(r;z,\hat{z})+\sup_{u\in\R}\Ex\left[\max_{s\in[t,T]}|\partial_z M^{t,u}_s|^{2p}\right]\Ex\left[\left|M_r^{t,z}-M_r^{t,\hat{z}}\right|^{2p}\right]\right\}dr\notag\\
  &\leq C\int_t^s\left\{I(r;z,\hat{z})+\left|z-\hat{z}\right|^{2p}\right\}dr.
\end{align*}
The Gronwall's lemma yields that, for all $s\in[t,T]$,
\begin{align*}
  \sup_{\hat{z}\neq z}\Ex\left[\sup_{r\in[t,s]}\left|\frac{\partial_zM_r^{t,z}-\partial_zM_r^{t,\hat{z}}}{z-\hat{z}}\right|^{2p}\right]<+\infty,
\end{align*}
and hence the left hand side of \eqref{eq:uniformly-integrable} with $\hat{z}\neq z$ and $s\in[t,T]$ is uniformly integrable. Then
\begin{align}\label{eq:hyy}
  h_{zz}(t,z,u)&=-\lim_{\hat{z}\rightarrow z}\Ex\left[\int_t^T e^{-\rho s-R_s^{t,u}}\frac{f'(s,M_s^{t,z})\partial_zM_s^{t,z}-f'(s,M_s^{t,\hat{z}})\partial_zM_s^{t,\hat{z}}}{z-\hat{z}}ds\right]\notag\\
  &=-\Ex\left[\int_t^T e^{-\rho s-R_s^{t,u}}\left(f''(s,M_s^{t,z})\left|\partial_zM_s^{t,z}\right|^2+f'(s,M_s^{t,z})\partial^2_{zz}M_s^{t,z}\right)ds\right],
\end{align}
where $f''(t,z)$ denotes the second-order partial derivative of $f$ w.r.t. $z$.

We then move on to the expression of $h_t$. Let us consider the solutions $M^{t,y}=(M_s^{t,z})_{s\in[t,T]}$ and $M^{\hat{t},z}=(M^{\hat{t},z}_s)_{s\in[\hat{t},T]}$ of SDE~\eqref{eq:Feynman-Kac-Y} with parameters $(t,z)\in[0,T]\times\R$ and $(\hat{t},z)\in[0,T]\times\R$ respectively. Moreover, for $r\geq0$, we introduce $\F_r^t:=\F_{t+r}$, $B^{1,t}_r:=B^1_{t+r}-B^1_t$, and $B^{2,t}_r:=B^2_{t+r}-B^2_t$ and define $\F^{\hat{t}}_r$, $B^{i,\hat{t}}_r$, $i=1,2$ for $r\geq0$ in a similar way. It is easy to check that
\begin{align}\label{eq:laweqn}
(M^{t,z}_{{t}+r},B^{1,t}_{r},B_r^{2,t})_{r\geq0}\overset{d}{=}(M^{{\hat{t}},z}_{{\hat{t}}+r},B^{1,{\hat{t}}}_{r},B_r^{2,\hat{t}})_{r\geq0}.
\end{align}
For any $\delta\in[0,T-t]$, it holds that
\begin{align*}
h(t+\delta,z,u)&= -\Ex\left[\int_{t+\delta}^T e^{-\rho s}f(s,M_s^{t+\delta,z})e^{-R_s^{t+\delta,u}}ds\right]=-\Ex\left[e^{-\rho\delta}\int_t^{T-\delta} e^{-\rho s}f(s,M_s^{t,z})e^{-R_s^{t,u}}ds\right].
\end{align*}
The dominated convergence theorem gives that
\begin{align*}
&\quad\lim_{\delta\downarrow0}\frac1\delta(h(t+\delta,z,u)-h(t,z,u))\notag\\
&=\lim_{\delta\downarrow0}\Ex\left[\frac{e^{-\rho\delta}}\delta\int_{T-\delta}^T e^{-\rho s}f(s,M_s^{t,z})e^{-R_s^{t,u}}ds+\frac{1-e^{-\rho\delta}}{\delta}\int_t^Te^{-\rho s}f(s,M_s^{t,z})e^{-R_s^{t,u}}ds\right]\notag\\
&=\Ex\left[e^{-\rho T}f(T,M_T^{t,z})e^{-R_T^{t,u}}+\rho\int_t^Te^{-\rho s}f(s,M_s^{t,z})e^{-R_s^{t,u}}ds\right].
\end{align*}
Similarly, for $t\in(0,T]$, we have that
\begin{align*}
\lim_{\delta\downarrow0}\frac1\delta(h(t,z,u)-h(t-\delta,z,u))
  &=\Ex\left[e^{-\rho T}f(T,M_T^{t,z})e^{-R_T^{t,u}}+\rho\int_t^Te^{-\rho s}f(s,M_s^{t,z})e^{-R_s^{t,u}}ds\right].
\end{align*}
Therefore, we conclude that, for $(t,z,u)\in\mathcal{D}_T$,
\begin{align}\label{eq:ht}
  h_t(t,z,u)=\Ex\left[e^{-\rho T}f(T,M_T^{t,z})e^{-R_T^{t,u}}+\rho\int_t^Te^{-\rho s}f(s,M_s^{t,z})e^{-R_s^{t,u}}ds\right].
\end{align}

At last, we verify the continuity of $h_t$, $h_{uu}$, $h_{yu}$ and $h_{yy}$ in $(t,y,u)$ using expressions \eqref{eq:ht}, \eqref{eq:Rephuu}, \eqref{eq:hyu} and \eqref{eq:hyy}. In fact, by Theorem 3.4.3 of \cite{Kunita19}, we have that $M^{t,z}_s$, $\partial_zM^{t,z}_s$ and $\partial^2_{zz}M^{t,z}_s$ for $s\in[t,T]$ admit the respective modifications which are continuous in $(t,z,s)$, $\Px$-a.s.. Moreover, by \eqref{eq:driftBM} and \eqref{eq:Lskorokhodrep},  $R_s^{t,u}$ is also continuous in $(t,u,s)$, $\Px$-a.s.. We can conclude by the dominated convergence theorem that $h_t$, $h_{zu}$ and $h_{zz}$ are continuous in $(t,z,u)$. For the continuity of $h_{uu}$, in view of \eqref{eq:laweqn}, for any $\epsilon\in\R$ satisfying $\tau^\epsilon_0:=\tau_{u_0}^t\wedge (T-\epsilon)\in[t,T]$,
\begin{align}\label{eq:Rephuu2}
h_{uu}(t+\epsilon,z,u_0)&=\Ex\left[e^{-\rho \tau^\epsilon_0}f(\tau^\epsilon_0,M_{\tau^\epsilon_0}^{t,z})\Gamma(\tau^\epsilon_0)\right]
+\Ex\left[\int_t^{\tau^\epsilon_0}e^{-\rho s}f(s,M_s^{t,z})e^{-R_s^{t,u_0}}ds\right].
\end{align}
Note that $\tau^\epsilon_0$ is continuous in $(u_0,\epsilon)$, $\Px$-a.s.. The continuity of $h_{uu}$ in $(t,z,u)$ follows from \eqref{eq:Rephuu2} and the dominated convergence theorem, which completes the entire proof.
\end{proof}

\begin{proof}[Proof of Theorem \ref{thm:hsolvePDE}]
For $(t,u)\in[0,T]\times[0,\infty)$, we define $B^{t,u}_s:=u + \sqrt{2\alpha}(B^1_s-B^2_t)+(\alpha-\rho)(s-t)$ for $s\in[t,T]$ and $M^{t,z}=(M_s^{t,z})_{s\in[t,T]}$ for $(t,z)\in[0,T]\times\R$ is the unique (strong) solution of SDE~\eqref{eq:Feynman-Kac-Y} under the assumption {\AZ}. By Remark 4.17 of Chapter 5 in \cite{KaraShreve91}, the assumption {\AZ} guarantees that the time-homogeneous martingale problem on $(M^{t,z},B^{t,u})=(M_s^{t,z},B_s^{t,u})_{s\in[t,T]}$ is well posed.\footnote{The definition of well-posedness of a time-homogeneous martingale problem can be found in Definition~4.15 of Chapter 5 in \cite{KaraShreve91}, page~320.} By applying Theorem 5.4.20 in \cite{KaraShreve91}, $(M^{t,z},B^{t,u})$ is a strong Markov process with $(t,z,u)\in[0,T]\times\R\times\R_+$. For $\varepsilon\in(0,u)$, let us define that
\begin{align}\label{eq:tauepsilon}
  \tau_{\varepsilon}^t:=\inf\left\{s\geq t;~\left|B_s^{t,u}-u\right|\geq\varepsilon~\text{or}~\left|M_s^{t,z}-z\right|\geq\varepsilon\right\}\wedge T.
\end{align}
Because the paths of $(M^{t,z},B^{t,u})$ are continuous, we have $\tau_{\varepsilon}^t>t$, $\Px$-a.s.. For any $\hat{t}\in[t,T]$, it holds that
\begin{align}\label{eq:equBR}
B^{t,u}_{\hat{t}\wedge\tau_{\varepsilon}^t}=R^{t,u}_{\hat{t}\wedge\tau_{\varepsilon}^t}.
\end{align}
In fact, if $\tau_{\varepsilon}^t\leq\hat{t}$, the following two cases may happen:
\begin{itemize}
\item[(i)] $|B_{\tau_{\varepsilon}^t}^{t,u}-u|<\varepsilon$: this yields that $0<-\varepsilon+u<B^{t,u}_{\tau_{\varepsilon}^t}<\varepsilon+u$, and hence \eqref{eq:equBR} holds.
 \item[(ii)]  $|B_{\tau_{\varepsilon}^t}^{t,u}-u|\geq\varepsilon$: this implies that $B_{\tau_{\varepsilon}^t}^{t,u}= u+\varepsilon>0$ or $B_{\tau_{\varepsilon}^t}^{t,u}= u-\varepsilon>0$, and again \eqref{eq:equBR} holds.
\end{itemize}
If $\hat{t}<\tau_{\varepsilon}^t$, then $|B_{\hat{t}}^{t,u}-u|<\varepsilon$ and $|M_{\hat{t}}^{t,z}-z|<\varepsilon$. This implies that $0<-\varepsilon+u<B^{t,u}_{\hat{t}}<\varepsilon+u$, and hence \eqref{eq:equBR} holds. By \eqref{eq:equBR} and the strong Markov property, we get that
\begin{align*}
-\Ex\left[\int_{\hat{t}\wedge\tau_{\varepsilon}^t}^Tf(s,M_s^{t,z})e^{-R_s^{t,u}-\rho s}ds\bigg|\F_{\hat{t}\wedge\tau_{\varepsilon}^t}\right]=
  h\left(\hat{t}\wedge\tau_{\varepsilon}^t,M^{t,z}_{\hat{t}\wedge\tau_{\varepsilon}^t},B^{t,u}_{\hat{t}\wedge\tau_{\varepsilon}^t}\right),
\end{align*}
where the function $h$ is given by $h(t,z,u) = -\Ex\left[\int_t^T e^{-\rho s}f(s,M_s^{t,z})e^{-B_s^{t,u}-L_s^{t,R}}ds\right]$ in view of \eqref{eq:dualtransfromh}.

Therefore, for $(t,y,u)\in\mathcal{D}_T$, it holds that
\begin{align}\label{eq:DPP}
h(t,z,u)= \Ex\left[h\left(\hat{t}\wedge\taue,M^{t,z}_{\hat{t}\wedge\taue},B^{t,u}_{\hat{t}\wedge\taue}\right)
-\int_t^{\hat{t}\wedge\taue}e^{-\rho s}f(s,M_s^{t,z})e^{-B_s^{t,u}-L_s^{t,R}}ds\right].
\end{align}
By Proposition~\ref{prop:smoothh} and It\^o's formula, we have that
\begin{align}\label{eq:DPP1}
&\frac1{\hat{t}-t}\Ex\left[\int_t^{\hat{t}\wedge\taue}e^{-\rho s}f(s,M_s^{t,z})e^{-B_s^{t,u}-L_s^{t,R}}ds\right]
=\frac1{\hat{t}-t}\Ex\left[h\left(\hat{t}\wedge\taue,M^{t,z}_{\hat{t}\wedge\taue},B^{t,u}_{\hat{t}\wedge\taue}\right)-h(t,z,u)\right]\notag\\
  &=\frac1{\hat{t}-t}\Ex\left[\int_t^{\hat{t}\wedge\taue}(h_t+{\cal L}h)\left(s,M^{t,z}_s,B^{t,u}_s\right)ds\right]+\frac1{\hat{t}-t}\Ex\left[\int_t^{\hat{t}\wedge\taue}h_u\left(s,M^{t,z}_s,B^{t,u}_s\right)dL_s^{t,R}\right],
\end{align}
where the operator ${\cal L}$ acted on $C^2(\R\times[0,\infty))$ is defined for $g\in C^2(\R\times[0,\infty))$ that
\begin{align}\label{eq:operatorL}
  {\cal L}g:=\alpha g_{uu}+(\alpha-\rho)g_u+\phi(y)g_{uz}+\mu_Z(z)g_{z}+\frac{\sigma_Z^2(z)}{2}g_{zz}.
\end{align}
By \eqref{eq:tauepsilon}, the assumption {\Af} implies that $(e^{-\rho s}f(s,M_s^{t,z})e^{-B_s^{t,u}-L_s^{t,R}})_{s\in[t,\hat{t}\wedge\taue]}$ is bounded. The bounded convergence theorem yields that
\begin{align*}
  \lim_{\hat{t}\downarrow t}\frac1{\hat{t}-t}\Ex\left[\int_t^{\hat{t}\wedge\taue}e^{-\rho s}f(s,M_s^{t,z})e^{-B_s^{t,u}-L_s^{t,R}}ds\right]=f(t,z)e^{-u-\rho t}.
\end{align*}
Similarly, we have that
\begin{align*}
  \lim_{\hat{t}\downarrow t}\frac1{\hat{t}-t}\Ex\left[\int_t^{\hat{t}\wedge\taue}(h_t+{\cal L}h)\left(s,M^{t,z}_s,B^{t,u}_s\right)ds\right]=(h_t+{\cal L}h)\left(t,z,u\right).
\end{align*}
Note that $R_s^{t,u}>0$ on $s\in[t,\hat{t}\wedge\taue]$ for all $(t,u)\in[0,T]\times[0,\infty)$. We then have that
\begin{align*}
  \frac1{\hat{t}-t}\Ex\left[\int_t^{\hat{t}\wedge\taue}h_u\left(s,M^{t,z}_s,B^{t,u}_s\right)dL_s^{t,R}\right]=0.
\end{align*}
By applying \eqref{eq:DPP1}, we obtain that $(h_t+{\cal L}h)(t,z,u)=f(t,z)e^{-u-\rho t}$ on $(t,z,u)\in[0,T)\times\R\times\R_+$.

We next verify that the function $h$ in \eqref{eq:dualtransfromh} satisfies the boundary conditions of the Neumann problem~\eqref{eq:hHJB2}. By the representation form \eqref{eq:dualtransfromh}, it is easy to see that $h(T,z,u)=0$ for all $(z,u)\in\R\times[0,\infty)$. It remains to show the validity of homogeneous Neumann boundary condition. In fact, for $s\in[t,T]$, we have that, for any positive sequence $(u_n)_{n\geq1}$ satisfying $u_n\downarrow0$ as $n\to\infty$,
\begin{align}\label{eq:defAsu2}
  \bigcup_{n\geq1} A_s^{t,u_n}:=\bigcup_{n\geq1}\left\{\max_{r\in[t,s]}\left[-\sqrt{2\alpha}(B^1_r-B_t^1)-(\alpha-\rho)(r-t)\right]>u_n\right\}\in\F_s,\quad\text{and}
  ~~\Px\left(\bigcup_{n\geq1}A_s^{t,u_n}\right)=1.
\end{align}
In view of \eqref{eq:derihu} in Proposition~\ref{prop:smoothh} and the assumption {\AZ}, it follows by the dominated convergence theorem that
\begin{align}\label{eq:huty0proof}
 h_u(t,z,0)=\lim_{n\to\infty}\int_t^T\Ex\left[f(s,M_s^{t,z})e^{-R_s^{t,u}-\rho s}\mathbf{1}_{\left\{(A_s^{t,u_n})^c\right\}}\right]ds=0,\ \ (t,z)\in[0,T]\times\R.
\end{align}
That is, the Neumann boundary condition in \eqref{eq:hHJB2} holds.

We next assume that the Neumann problem~\eqref{eq:hHJB2} admits a classical solution $h$ with a polynomial growth. For $n\in\N$ and $t\in[0,T]$, we define $\tau_n^t:=\inf\{s\geq t;~|M_s^{t,z}|\geq n~\text{or}~|R_s^{t,u}|\geq n\}\wedge T$. It\^{o}'s formula gives that, for $(t,z,u)\in\mathcal{D}_T$,
\begin{align}\label{eq:unique-solution}
&\Ex\left[h\left({\tau_n^t},M^{t,z}_{\tau_n^t},R^{t,u}_{\tau_n^t}\right)\right]=h(t,z,u)+\Ex\left[\int_t^{\tau_n^t}(h_t+{\cal L}h)\left(r,M^{t,z}_r,R^{t,u}_r\right)dr\right]\\
  &\quad+\Ex\left[\int_t^{\tau_n^t}h_u(r,M^{t,z}_r,R^{t,u}_r)\mathbf{1}_{\{R^{t,u}_r=0\}}dL^{t,R}_r\right]
  =h(t,z,u)+\mathbb{E}\left[\int_t^{\tau_n^t}f(r,M^{t,z}_r)e^{-R^{t,u}_r-\rho r}dr\right].\notag
\end{align}
Moreover, the polynomial growth of $h$ implies the existence of a constant $C=C_T>0$ such that, for some $p\geq1$,
\begin{align*}
|h({\tau_n^t},M^{t,z}_{\tau_n^t},R^{t,u}_{\tau_n^t})|\leq C\big\{1+\max_{r\in[t,T]}|M^{t,z}_r|^p+\max_{r\in[t,T]}|R^{t,u}_r|^p\big\}.
\end{align*}
Note that $\lim_{n\rightarrow\infty}h({\tau_n^t},M^{t,z}_{\tau_n^t},R^{t,u}_{\tau_n^t})=h(T,M^{t,z}_T,R^{t,u}_T)=0$ in view of \eqref{eq:hHJB2}.  Letting $n\to\infty$ on both sides of \eqref{eq:unique-solution}, by dominated convergence theorem and monotone convergence theorem, we obtain the representation~\eqref{eq:dualtransfromh} of the solution $h$ that completes the proof.
\end{proof}

\begin{proof}[Proof of Corollary \ref{coro:wellposehatv}]
We first show the existence of a classical solution to the Neumann boundary problem \eqref{eq:hatvHJB1}. By Theorem~\ref{thm:hsolvePDE}, the function $h$ defined by \eqref{eq:dualtransfromh} solves the problem \eqref{eq:hHJB2}. It readily follows that for $(t,z,y)\in[0,T]\times\R\times(0,1]$, $\hat{v}(t,z,y):=e^{\rho t}h(t,z,-\ln y)$ solves the Neumann boundary problem \eqref{eq:hatvHJB1}. The existence of a classical solution to the problem \eqref{eq:hatvHJB1} then follows by Proposition~\ref{prop:smoothh}. For the uniqueness, let $\hat{v}^{(i)}$ for $i=1,2$ be two classical solutions of the Neumann problem \eqref{eq:hatvHJB1} such that $h^{(i)}(t,z,u):=e^{-\rho t}\hat{v}^{(i)}(t,z,e^{-u})$ for $(t,z,u)\in\mathcal{D}_T$ satisfies the polynomial growth for $i=1,2$. Theorem~\ref{thm:hsolvePDE} implies that both $h^{(1)}$ and $h^{(2)}$ admit the probabilistic representation \eqref{eq:dualtransfromh}, and hence $h^1=h^2$ on $\mathcal{D}_T$. Therefore, it holds that $\hat{v}^{(1)}(t,z,y)=e^{\rho t}h^{(1)}(t,z,-\ln y)=e^{\rho t}h^{(2)}(t,z,-\ln y)=\hat{v}^{(2)}(t,z,y)$ for $(t,z,y)\in[0,T]\times\R\times(0,1]$. Moreover, the strict convexity of  $(0,1]\ni y\to\hat{v}(t,z,y)$ for fixed $(t,z)\in[0,T)\times\R$ follows from the fact that $\hat{v}_{yy}=\frac{e^{\rho t}}{y^2}[h_{uu}+h_u]>0$ in Remark~\ref{rem:hu+hpositive} as it is assumed that $f(\cdot, \cdot)>0$.
\end{proof}

\section{Conclusions}\label{sec:conc}
This paper studies an optimal tracking problem when the capital injection is allowed such that the total fund capital needs to dynamically dominate a non-decreasing benchmark process. The aim of the stochastic control problem is to minimize the total capital injection. To cope with this stochastic control problem with American type floor constraints, we perform a 2-step transformation. We first transform the original problem into an unconstrained stochastic control problem with a running maximum cost and then transform the problem again to an auxiliary problem by working with a new state process with the reflection. The associated HJB equation has a Neumann boundary condition. By using the dual transform and the assumption that $f\geq 0$ in the benchmark process, we can establish the probabilistic representation of the solution to the dual PDE so that the regularity of the solution can be proved by means of stochastic flow analysis. As the original value function is not strictly concave, the verification theorem is carefully proved that gives the full characterization of the value function and the optimal portfolio. To exemplify the application of our theoretical results, we also discuss some market index tracking problems that can be transformed into problems with non-decreasing benchmarks.

Future extensions can be conducted in different directions. Firstly, it will be interesting to consider both capital injection and capital withdrawal in the tracking procedure, which lead to the two dimensional singular control policies that characterize the accumulated amount in the injected capital and the withdrawal of capital respectively. The tracking problem can be formulated to minimize the difference between the total capital injection and capital withdrawal when the withdrawal process subjects to some transaction costs. Secondly, one can consider some more general stochastic benchmark processes or some incomplete market models. The linearization by the dual transform may no longer work and we need to study the primal nonlinear HJB equation directly. At last, it is an appealing future work to investigate the duality theory for the optimal tracking problem in \eqref{eq_prob_IBP}, which may provide an alternative method to study the optimal tracking problem in general market models.

\ \\
\textbf{Acknowledgements}: \small{We would like to thank Martin Larsson and Johannes Ruf for their insightful discussions on some general index tracking problems during their visit to The Hong Kong Polytechnic University in 2018, who initially formulated the new tracking procedure using capital injection and dynamic floor constraints. We also thank two anonymous referees for their helpful comments on the presentation of this paper. L. Bo is supported by Natural Science Foundation of China under grant no. 11971368, no. 11961141009 and the Key Research Program of Frontier Sciences of the Chinese Academy of Science under grant no. QYZDB-SSW-SYS009. H. Liao is supported by Singapore MOE AcRF Grants R-146-000-271-112. X. Yu is supported by the Hong Kong Early Career Scheme under grant no. 25302116.}

\ \\

\end{document}